\documentclass[11pt,a4paper,reqno]{amsart}
\usepackage{amsmath,amssymb,amsthm,amsfonts}
\usepackage[export]{adjustbox}
\usepackage{amsbsy}
\usepackage[utf8]{inputenc}
\usepackage[normalem]{ulem}
\usepackage{xspace}
\usepackage{cancel}
\usepackage[greek,english]{babel}
\usepackage{url}
\usepackage{wrapfig}
\usepackage{enumitem}
\usepackage{multicol}
\usepackage{moreenum}
\usepackage{xcolor}
\usepackage{longtable}
\usepackage{array}

\usepackage[pagewise]{lineno} 
\newcounter{dummy}
\makeatletter
\newcommand\myitem[1][]{\item[#1]\refstepcounter{dummy}\def\@currentlabel{#1}}
\makeatother
\usepackage{subcaption}
\usepackage{comment}
\usepackage{tikz}
\usetikzlibrary{backgrounds}
\usepackage[bottom=2.5cm,top=2.5cm, left=2.5cm, right=2.5cm]{geometry}

\usepackage{multirow}
\usepackage{array,multirow}

\definecolor{LinkColor}{rgb}{0,0,0} 
\usepackage[colorlinks=true,linkcolor=LinkColor,citecolor=LinkColor,urlcolor= LinkColor, naturalnames, hyperindex, pdfstartview=FitH, bookmarksnumbered, plainpages]{hyperref}
\usepackage{cleveref}

\makeatletter
\newcommand{\slunlhd}{%
	\mathrel{\mathpalette\sl@unlhd\relax}%
}

\newcommand{\sl@unlhd}[2]{%
	\sbox\z@{$#1\lhd$}%
	\sbox\tw@{$#1\leqslant$}%
	\dimen@=\ht\tw@
	\advance\dimen@-\ht\z@
	\ifx#1\displaystyle
	\advance\dimen@ .2pt
	\else
	\ifx#1\textstyle
	\advance\dimen@ .2pt
	\fi
	\fi
	\ooalign{\raisebox{\dimen@}{$\m@th#1\lhd$}\cr$\m@th#1\leqslant$\cr}%
}
\makeatother

\newtheorem{innercustomthm}{Theorem}[section]

\crefname{innercustomthm}{Theorem}{Theorems}

\newtheorem{theorem}{Theorem}[section]

\newtheorem{lemma}[theorem]{Lemma}

\newtheorem{proposition}[theorem]{Proposition}

\theoremstyle{definition}
\newtheorem{definition}[theorem]{Definition}
\newtheorem{remark}[theorem]{Remark}

\author[S. Chahal]{Seema Chahal}
\address{(Seema Chahal) Department of Mathematics, Indian Institute of Technology Roorkee, Roorkee (Uttarakhand)-247667, India.}
\email{\href{mailto:seema_r@ma.iitr.ac.in}{seema\_r@ma.iitr.ac.in}}
\author[S. Antil]{Seema Antil}
\address{(Seema Antil) Department of Mathematics, 
	Indian Institute of Technology Ropar, Ropar (Punjab)-140001, India.}
\email{\href{mailto:seema.22maz0010@iitrpr.ac.in}{seema.22maz0010@iitrpr.ac.in}}
\author[S. Maheshwary]{Sugandha Maheshwary}
\address{(Sugandha Maheshwary) Indian Institute of Technology Roorkee, Roorkee (Uttarakhand)-247667, India.}
\email{\href{mailto:msugandha@ma.iitr.ac.in}{msugandha@ma.iitr.ac.in}}
\author[M. Khan]{Manju Khan}
\address{(Manju Khan) Department of Mathematics, 
	Indian Institute of Technology Ropar, Ropar (Punjab)-140001, India.}
\email{\href{mailto:manju@iitrpr.ac.in}{manju@iitrpr.ac.in}}

\thanks{The second-named author gratefully acknowledges the financial support provided by the Ministry of Eduction, Government of India. The second author also acknowledges the partial support from the FIST program of the Department of Science and Technology, Government of India, Reference No. SR/FST/MS-I/2018/22(C). The third-named author gratefully acknowledges the support by Science  \& Engineering Research Board (SERB),  DST (Department of Science and Technology), India (SRG/2023/000180), The fourth-named author acknowledges the financial support of ANRF (File no. MTR/2022/000616).}


\date{}

\title{Skew Constacyclic Codes Of Length $np^s$ over $ \frac{\mathbb{F}_{p^m}[u]}{\langle u^k \rangle}$ }
\begin{document}
	\maketitle
	\begin{abstract}
		Let $\mathbb{F}_{p^m}$ be the field containing $p^m$ elements where $p$ is an odd prime and $m \in \mathbb{N}$.
		In this article, we propose a unified approach to the study of skew constacyclic codes of length $np^s$ over the ring $R_k = \mathbb{F}_{p^m}[u]/\langle u^k \rangle,$ where $n, s, k \in \mathbb{N}$  and $\gcd(n, p)=1$. 
		Consider the skew polynomial ring $R_k[x;\Theta]$, where
		$\Theta$ is an automorphism of $R_k$ such that $xa = \Theta(a)x$ for all $a \in R_k$. 	Let $f(x)$ be a central irreducible divisor of $x^{np^s} - \lambda$ of degree $l$ and multiplicity $j$ in  $R_k[x;\Theta]$, where $\lambda $ is an invertible element in $R_k$.	In this article, we study skew constacyclic codes of length \(np^s\) over \(R_k\), which reduces to the study of skew polycyclic codes of length $jl$  associated with a polynomial \(f(x)^j\).
		 Using the fact that skew polycyclic codes associated with a polynomial \(f(x)^j\)  can be described by the left ideal structure of the quotient ring $R_k[x;\Theta]/\langle f(x)^{j}\rangle$, we investigate this class of codes for specific choices of  $\Theta$. In particular, if $\lambda$ is an invertible element of $\mathbb{F}_{p^m}$, we classify all left ideals and establish an isomorphism between skew cyclic and skew constacyclic codes,  under suitable conditions.
		  Furthermore,  we provide a comprehensive analysis of skew constacyclic codes of length $3p^s$ over $R_k$. Finally, we examine skew cyclic and skew negacyclic codes of length $6p^s$ over $R_k$ using the factorization of  $x^{6p^s} - 1$ and $x^{6p^s} + 1$, respectively; with a complete case-by-case analysis.  Examples demonstrating codes with optimal parameters are also included.
	\end{abstract}

		\section{Introduction}
	The class of cyclic codes plays a central role in the theory of error-correcting codes. As a generalization of cyclic codes, constacyclic codes have been extensively studied. In recent years, there has been growing interest in constacyclic codes over various rings.  Though, for specific lengths and over particular chain  rings, some classification of constacyclic codes have been undertaken, a complete classification of constacyclic codes with respect to arbitrary lengths and chain rings remains challenging in general.
	 	In recent years, much work has been done for constacyclic codes over finite chain rings and also it has been proven that constacyclic codes over finite chain ring has many practical applications. In this direction a prominently used ring $R_k = \mathbb{F}_{p^m}[u]/\langle u^k \rangle
	= \mathbb{F}_{p^m} + u\mathbb{F}_{p^m} + \cdots + u^{k-1}\mathbb{F}_{p^m}$.
	The structure of constacyclic codes over $\mathbb{F}_{p^m}[u]/\langle u^k \rangle$ has been extensively studied in the literature. Particularly, for  $k=2$, significant attention has been devoted to cyclic and constacyclic codes of length $n$ over the ring $\mathbb{F}_{p^m} + u\mathbb{F}_{p^m}$ for various primes $p$ and positive integer $n$ (see, e.g., \cite{Din10, Din13, Din14, CDL16, DDS17}). In this direction
 Zhao et al. \cite{ZTG18}  and Cao et al. \cite{CCDF18} studied constacyclic codes of length $np^s$ over $\mathbb{F}_{p^m} + u\mathbb{F}_{p^m}$.	More recently, Seema et al. \cite{AKK25} investigated cyclic codes of length $n$ over finite chain rings. 
		Beyond the commutative setting, Boucher et al. \cite{BGU07} initiated the study of skew cyclic codes over the non-commutative skew polynomial ring $\mathbb{F}_{p^m}[x;\theta]$, where $\theta$ is an automorphism of the finite field $\mathbb{F}_{p^m}$. Subsequently, Jitman et al. \cite{JLU12} investigated the algebraic structure and fundamental properties of skew constacyclic codes of length $n$ over finite chain rings. 
		Later, Bagheri et al. \cite{BHRS22} classified all skew cyclic codes of length $p^s$ over $\mathbb{F}_{p^m} + u\mathbb{F}_{p^m}.$  More recently, Hesari et al. \cite{HS23} analyzed the algebraic structure of skew constacyclic codes of lengths $p^s$ and $2p^s$ over $\mathbb{F}_{p^m} + u\mathbb{F}_{p^m}$.  In \cite{SRD25}, the authors determined the algebraic structure of skew negacyclic codes of length $4p^s$ and characterized their self-dual codes. In \cite{HS25}, the authors studied skew constacyclic codes of length $p^s$ over $\mathbb{F}_{p^m} + u\mathbb{F}_{p^m} + u^2\mathbb{F}_{p^m}$ with $u^3=0$. Recently,  Bajalan et al. \cite{BMO26} studied skew polycyclic codes over finite chain rings.
		Motivated by these works, in this paper we develop a uniform framework for the study of skew constacyclic codes over $\mathbb{F}_{p^m}[u]/\langle u^k \rangle$. 
	For notational convenience, denote $ \mathbb{F}_{p^m}[u]/\langle u^k \rangle$  by $R_k$ and  $R_k[x;\Theta]$ is  the skew polynomial ring, with an automorphism $\Theta$ of $R_k$  satsifying $xa = \Theta(a)x$ for all $a \in R_k$. For a central irreducible divisor $f(x)$ of degree $l$ and multiplicity $j$ of $x^{np^s}-\lambda$ in $R_k[x;\Theta]$, where $\lambda$ is a unit in $R_k$, $p$ is an odd prime, and $\gcd(n,p)=1, j,l,n \in \mathbb{N}$. We study skew constacyclic codes of length $np^s$ over $R_k$, reducing the problem to the study of skew polycyclic codes of length $jl$ associated with a   polynomial $f(x)^j$. 
	We investigate skew  $(f^j,\Theta)$-polycyclic codes, which correspond to left ideals of the quotient ring
	$R_k[x;\Theta]/\langle f(x)^{j}\rangle, ~ j\in\mathbb{N}.$
		 In particular, for $\lambda=\alpha \in \mathbb{F}_{p^m}^{*}$, we characterize all types of left ideals of $R_k[x;\Theta]/\langle f(x)^{j} \rangle, ~ j \in \mathbb{N}$ and establish an isomorphism between skew $\Theta$-cyclic codes and skew $(\lambda,\Theta)$-constacyclic codes under suitable conditions on $\lambda$ and $\Theta$.	
	As special cases, when $\Theta$ is the identity automorphism, we recover many known results on cyclic and negacyclic codes of length $np^s$ over $R_k$. Moreover, for $\lambda=\alpha^3$ with $\alpha\in\mathbb{F}_{p^m}^{*}$, we show that every skew $(\lambda,\Theta)$-constacyclic code of length $3p^s$ is isomorphic to a skew $\Theta$-cyclic code. We also study skew $(\lambda,\Theta)$-constacyclic codes over $\mathbb{F}_{p^m}$ when $\lambda$ is not a cube in $\mathbb{F}_{p^m}^{*}$, and provide a detailed analysis of skew constacyclic codes of length $3p^s$ over $R_k$, where $k>1$. We also  study skew $\Theta$-cyclic and skew $\Theta$-negacyclic codes of length $6p^s$ over $R_k$ using the factorizations of $x^{6p^s}-1$ and $x^{6p^s}+1$, respectively. Furthermore, we provide several examples of MDS codes to illustrate the obtained results.
		This paper is organized as follows. In Section~2, we present basic preliminary definitions, notation, and some straightforward results related to our work. Section~3 is devoted to the study of $(f^j, \Theta)$-polycyclic codes of length $jl$ over $R_k$, $j,l \in \mathbb{N}$. In particular, for $\lambda=\alpha \in \mathbb{F}_{p^m}^{*}$, we characterize all types of left ideals of $R_k[x;\Theta]/\langle f(x)^{j} \rangle$ and establish an isomorphism between skew $\Theta$-cyclic codes and skew $(\lambda,\Theta)$-constacyclic codes under suitable conditions on $\lambda$ and $\Theta$.	
		A comprehensive analysis of skew constacyclic codes of length $3p^s$ over $R_k$, is presented in Section~4.
	 In Section~5, we study skew $\Theta$-cyclic and skew $\Theta$-negacyclic codes of length $6p^s$ over $R_k$ using the factorizations of $x^{6p^s}-1$ and $x^{6p^s}+1$, respectively. Finally, several examples with optimal parameters are also provided to illustrate the applicability of the results.

	\section{Notation and Preliminaries}
As already mentioned,  $\mathbb{F}_{p^m}$, where $p$ is a prime and $m \in \mathbb{N}$, denotes the   finite field of order $p^m$ and for  $k \in \mathbb{N},$   ring $R_k$ is the finite chain ring:
\[ R_k = \frac{\mathbb{F}_{p^m}[u]}{\langle u^k \rangle} = \mathbb{F}_{p^m} + u\mathbb{F}_{p^m} + \dots + u^{k-1}\mathbb{F}_{p^m}, \]
where $u^k = 0$. The set of invertible elements of $\mathbb{F}_{p^m}$ and $R_k$ are denoted, respectively by $\mathbb{F}_{p^m}^*$ and $R_k^*$.

\begin{definition}
	Let $\Theta$ be an automorphism of $R_k$. The \emph{skew polynomial ring} $R_k[x;\Theta]$ consists of polynomials
	\[
	a_0 + a_1 x + \cdots + a_n x^n, \quad a_i \in R_k,\; n \in \mathbb{N} \cup \{0\},
	\]
	with the usual addition and multiplication determined by the rule
	\[
	x a = \Theta(a) x \quad \text{for all } a \in R_k,
	\]
	extended by associativity and distributivity. The set $R_k[x;\Theta]$ with above operations forms a ring called the skew polynomial ring over $R_k$ and every element in $R_k[x;\Theta]$  is called skew polynomial.
\end{definition}
We say that $f(x)$ is a right divisor (left divisor) of $g(x)$ in $R_k[x;\Theta]$, and we write $f(x) \mid_r g(x) \quad ( f(x) \mid_l g(x))$ if there exists a skew polynomial $h(x) \in R_k[x;\Theta]$ such that $g(x) = h(x)f(x) \quad ( g(x) = f(x)h(x)).$ 
\begin{definition}
	Suppose \(f(x), g(x)\) are skew polynomials in $R_k[x;\Theta]$. The greatest common right divisor of \(f(x)\) and \(g(x)\) is the monic polynomial $d_r(x) \in R_k[x;\Theta]$ such that
	$d_r(x)\mid_r f(x),$  $d_r(x)\mid_r g(x),$ and for any $d_r'(x)\in R_k[x;\Theta]$ satisfying
	$d_r'(x)\mid_r f(x)$ and $d_r'(x)\mid_r g(x),$	we have $d_r'(x)\mid_r d_r(x).$ We denote \(d_r(x)\) by $\gcd_r(f(x),g(x)).$
\end{definition}

The following lemma enables division in $R_k[x;\Theta]$.
\begin{lemma}\cite{McD74}Let $f(x), g(x) \in R_k[x;\Theta]$  be such that  $f (x)$ having invertible leading cofficeient. Then there exist $q(x), r(x) \in R_k[x;\Theta]$ with $g(x)=q(x)f(x)+r(x), $ where $r(x)=0$ or $deg(r(x))<deg(f(x)). $ 
\end{lemma}

\begin{proposition}\cite{JLU12}\label{Central poly}
	Let $n$ be a positive integer and let $\lambda$ be a unit in $R_k$. Then the following statements are equivalent:
	\begin{enumerate}
		\item $x^{n} - \lambda$ is central in $R_k[x;\Theta]$.
		\item $\langle x^{n} - \lambda \rangle$ is a two-sided ideal of $R_k[x;\Theta]$.
		\item $n$ is a multiple of the order of $\Theta$ and $\lambda$ is fixed by $\Theta$, that is, $\Theta(\lambda)=\lambda$.
	\end{enumerate}
\end{proposition}
\begin{proposition}\cite{BMO26}\label{Central poly generalization}
	A monic polynomial
	$f(x)=x^n-\sum_{i=0}^{n-1} a_i x^i \in R_k[x;\Theta]$
	is central if and only if
	\begin{enumerate}
		\item \( \Theta(a_i)=a_i \) for all \(0 \le i \le n-1\);
		\item \( a_i r-\Theta^i(r)a_i=0 \) for all \(r \in R_k\) and \(0 \le i \le n-1\);
		\item \( \Theta^n=\mathrm{Id}_{R_k}, \) where \(\mathrm{Id}_{R_k}\)
		denotes the identity automorphism of \(R_k\).
	\end{enumerate}
\end{proposition}
A code $\mathcal{C}$ of length $n$ over $R_k$ is a non-empty subset of $R_k^n$ and the  ring $R_k$ is referred to as the alphabet of $\mathcal{C}$ and the elements of $\mathcal{C}$ are called codeword. For any codeword $c \in \mathcal{C}$, the Hamming weight $w(c)$ of $c$ is defined as the number of non-zero components of $c$.  In addition, if $\mathcal{C}$  is  an $R_k$-submodule of $R_k^n$, then $\mathcal{C}$ is called a linear code. 
The Hamming distance, denoted by  $d$ of the linear code $\mathcal{C}$ is defined by $d=\min\{w(c):0\neq c\in \mathcal{C}\}.$ 
 The Euclidean dual of a linear code $\mathcal{C} \subseteq R_k^n$ is the linear code defined by 
$$\mathcal{C}^\perp=\{x \in R_k^n : x.y=0, ~\text{for all} ~y \in \mathcal{C}\},$$ where $x.y$ denotes the standard Euclidean inner product. If $\mathcal{C}=\mathcal{C}^\perp,$ then $\mathcal{C}$ is called a self-dual code.\\
For $\lambda \in R_k^*$, a linear code $\mathcal{C}$ over $R_k$ 
is called $\lambda$-constacyclic,   if $ (\lambda a_{n-1}, a_0, a_1, \ldots, a_{n-2})\in \mathcal{C}$ whenever $(a_0, a_1, \ldots, a_{n-1})\in \mathcal{C} $. 

In other words, a linear code over $R_k$ is   $\lambda$-constacyclic, if it is an ideal of $\frac{R_k[x]}{\langle  x^n- \lambda\rangle}$, where $\lambda \in R_k^{*}$.
Further, for a given  automorphism $\Theta$ of $R_k$ and  $\lambda \in R_k^*$, a linear code $\mathcal{C}$  is called  skew $(\lambda,\Theta)$-constacyclic, if $\mathcal{C}$ is closed under the $\Theta$-constacyclic shift $\rho_{\Theta} : R_k^n \rightarrow R_k^n$ defined by 
$\rho_{\Theta}((a_0,a_1,\ldots,a_{n-1}))=(\Theta(\lambda(a_{n-1})), \Theta(a_0), \Theta(a_1),\ldots, \Theta(a_{n-2}))$. Equivalently, a linear code over \(R_k\) is called a  skew $(\lambda, \Theta)$-constacyclic code if it is an left ideal of
$\frac{R_k[x, \Theta]}{\langle x^n-\lambda \rangle}.$
If $\mathcal{C}$ is skew $( \lambda, \Theta)$-constacyclic code over $R_k$, then $\mathcal{C}^{\perp} $ is a skew $( \lambda^{-1}, \Theta)$-constacyclic code over $R_k$.
Note that skew $( \lambda, \Theta)$-constacyclic code $\mathcal{C}$ is called skew $\Theta$-cyclic, if $\lambda=1$ and skew $\Theta$-negacyclic, if $\lambda=-1.$
 It is easy to observe that skew  $( \lambda, \Theta)$-constacyclic codes are $\lambda$-constacyclic, if $\Theta$ is identity automorphism of $R_k$ and in this case if $\lambda=1$ then $\mathcal{C}$
is a cyclic code.  Likewise, if $\lambda=-1,$ then $\mathcal{C}$ is a negacyclic code. 
More generally, skew polycyclic codes were introduced in the literature  \cite{BMO26}. Let
$
f(x)=x^n-\sum_{i=0}^{n-1} a_i x^i \in R_k[x;\Theta]
$
be a monic polynomial of degree \(n\). A linear code
\(\mathcal{C}\subseteq R_k^n\) is called a skew \((f, \Theta)\)-polycyclic code if for every
$
(c_0,c_1,\ldots,c_{n-1})\in \mathcal{C},
$
we have
$
\bigl(
\Theta(c_{n-1})a_0,\,
\Theta(c_0)+\Theta(c_{n-1})a_1,\,
\ldots,\,
\Theta(c_{n-2})+\Theta(c_{n-1})a_{n-1}
\bigr)\in \mathcal{C}.
 $ If $f(x)$ is central, then a linear code $\mathcal{C}$  over $R_k$ is skew $(f, \Theta)$-polycyclic code  if and only if it is an left ideal of $\frac{R_k[x;\Theta]}{\langle f(x) \rangle}.$
Observe that skew  $(\lambda, \Theta)$-constacyclic codes of length $np^s$ are precisely skew \((f, \Theta)\)-polycyclic codes corresponding to
$f(x)=x^{np^s}-\lambda,
~ \lambda\in R_k^{*}.$

Now consider $\mathcal{R}^{np^s}_{k, \lambda}=\frac{R_k[x, \Theta]}{\langle x^{np^s}-\lambda \rangle} ,$  where $\lambda \in R_k^*$.
	\begin{remark}\label{remark 1} If $f(x)=\sum_{i=0}^{r}b_ix^i  \in R_k[x,\Theta]$ and $b_r \neq 0, ~b_i \in R_k$, then the reciprocal of $f(x)$ is the polynomial 
			$f^*(x)=\sum_{i=0}^{r}\Theta(b_{r-i})x^i.$
				Let $\mathcal{I}$ be a left ideal of $\mathcal{R}_{k, \lambda}^{np^s}.$ Then, 
			\[
			\mathcal{A(I)}=\{f(x) \in \mathcal{R}_{k, \lambda}^{np^s} \mid g(x)f(x)=0, ~\forall~g(x)\in \mathcal{I}\}
			\]
			is a right ideal of $\mathcal{R}_{k, \lambda}^{np^s}$ and is called 
			the right annihilator of $\mathcal{I}.$
	 It is easy to verify that if $\mathcal{C}$ is a skew $(\lambda, \Theta)$-constacyclic code associated with the left ideal $\mathcal{I}$ of  $\mathcal{R}^{np^s}_{k, \lambda}$, then the  left ideal of  $\mathcal{R}^{np^s}_{k, \lambda}$ associated with $\mathcal{C}^\perp$ is $\mathcal{A(I)^*}$, where $\mathcal{A(I)^*}$ = $\{a^*(x)~|~ a(x) \in \mathcal{A}(\mathcal{I})\}.$
	\end{remark}

		Let $\Theta$ be an automorphism of $R_k$ and $\theta = \Theta|_{\mathbb{F}_{p^m}}$ be its restriction to the base field. 
	
The map $ \pi_{k}$,
	\[ \pi_{k}:R_{k}\rightarrow R_{k-1} \]
	\[ \sum_{j=0}^{k-1}u^{j}a_{j}\mapsto\sum_{j=0}^{k-2}u^{j}a_{j}, \] where $a_j \in F_{p^m}$, is a surjective ring homomorphism, which may be canonically extended to  as,
	\[ \pi_{k}: \mathcal{R}^{np^s}_{k, \lambda}\rightarrow \mathcal{R}^{np^s}_{k-1, \lambda} \]
	\[ \sum_{j=0}^{k-1}u^{j}g_{j}(x)\mapsto\sum_{j=0}^{k-2}u^{j}g_{j}(x), \]
	Likewise,  we have the  surjective ring homomorphism map $\mu_{k}$,
	\[ \mu_{k}: \mathcal{R}^{np^s}_{k, \lambda}\rightarrow  \mathcal{R}^{np^s}_{1, \alpha}=\frac{\mathbb{F}_{p^m}[x;  \theta]}{\langle x^{np^s}-\alpha \rangle}, \] 
	\[ \sum_{j=0}^{k-1}u^{j}g_{j}(x)\mapsto g_{0}(x),\] where $\alpha=\mu_k(\lambda).$
	
	\begin{definition}
		If $\mathcal{C}_1$ and $\mathcal{C}_2$ are left ideals of a ring  $\mathcal{R}^{np^s}_{k, \lambda}$, then their left ideal quotient is 
		\[ (\mathcal{C}_1 : \mathcal{C}_2)=\{ g(x)\in \mathcal{R}^{np^s}_{k, \lambda} : g(x)\mathcal{C}_2 \subseteq \mathcal{C}_1 \} \]
		which is also an left ideal. In particular, if $\mathcal{C}_2=\mathcal{R}^{np^s}_{k, \lambda}(u),$ we can write \[(\mathcal{C}_1 : u) =\{ g(x)\in \mathcal{R}^{np^s}_{k, \lambda} : u g(x)\in \mathcal{C}_1 \}\]
		\end{definition}
	
		\begin{definition}
		Let $\mathcal{C}$ be a skew $(\lambda,\Theta)$-constacyclic code of length $np^s$ over $R_k$.	For $i=1,2, \ldots, k$, the \emph{$i$-th torsion code} of $\mathcal{C}$ is the code over
		$\mathbb{F}_{p^m}$ defined by
		\[
		\operatorname{Tor}_{i}(\mathcal{C})
		= \mu_k\!\left((\mathcal{C} :_{\mathcal{R}^{np^s}_{k,\lambda}} u^{i-1})\right)
		= \mu_k\!\left(\{\, g(x)\in \mathcal{R}^{np^s}_{k,\lambda} \mid u^{i-1} g(x)\in \mathcal{C} \,\}\right),
		\]
		where $\mu_k$ is as defined above, which is 
  a left ideal of $\mathcal{R}^{np^s}_{1,\alpha}$.
	\end{definition}
Let $\Theta \in \mathrm{Aut}(R_k)$  such that the order of $\Theta$ divides $np^s$ and assume that $\lambda \in R_k^{*}$ such that $\Theta(\lambda)=\lambda$.   Hence by  $\Cref{Central poly generalization}$ we have $x^{np^s}-\lambda$ is central in $R_k[x;\Theta]$.
By using
\cite[Chapter II; Chapter XX, Exercise 7]{McD74}
, there exist irreducible  polynomials $f_1(x),  f_2(x),\ldots, f_r(x)$ in $R_k[x; \Theta]$ such that
 $
x^{np^s}-\lambda = f_1(x)f_2(x) \cdots f_r(x).$
Assume that each $f_j(x)$ is also central polynomial in $R_k[x;\Theta]$. Then we have, $
x^{np^s}-\lambda = f_1(x)^{k_1} f_2(x)^{k_2} \cdots f_t(x)^{k_t},$ where each $f_j(x)$ is irreducible and pairwise coprime.
For each $1 \le j \le t$, define $
F_j(x)=\frac{x^{np^s}-\lambda}{f_j(x)^{k_j}}.
$
Since $\gcd_r(f_j(x)^{k_j},F_j(x))=1$, there exist $v_j(x), w_j(x) \in R_k[x;\Theta]$ such that $$
v_j(x)F_j(x) + w_j(x)f_j(x)^{k_j} = 1.$$

Define
\begin{equation}
	\varepsilon_j(x) = v_j(x)F_j(x)\mod{x^{np^s}-\lambda}.
\end{equation}
Then we have the following results:

\begin{lemma}
	Let
	$
	\mathcal{R}^{np^s}_{k,\lambda}=\frac{R_k[x;\Theta]}{\langle x^{np^s}-\lambda\rangle}
	\quad \text{and} \quad 
	\mathcal{R}^{np^s}_{k,f_j, k_j}=\frac{R_k[x;\Theta]}{\langle f_j(x)^{k_j}\rangle}.
	$
	Then the following statements hold:
	\begin{enumerate}
		\item $\varepsilon_1(x) + \cdots + \varepsilon_t(x) = 1$, $\varepsilon_j(x)^2 = \varepsilon_j(x)$ and 
		$\varepsilon_j(x)\varepsilon_\ell(x) = 0$ in $\mathcal{R}^{np^s}_{k,\lambda}$ for all $1 \le j \neq l \le t$.
		
		\item $\mathcal{R}^{np^s}_{k,\lambda} = \mathcal{A}_1 \oplus \cdots \oplus \mathcal{A}_t$, where 
		$\mathcal{A}_j = \mathcal{R}^{np^s}_{k,\lambda}\varepsilon_j(x)$ with $\varepsilon_j(x)$ as its multiplicative identity, and 
		$\mathcal{A}_j \mathcal{A}_\ell = \{0\}$ for all $1 \le j \neq l \le t$.
		
		\item For each $1 \le j \le t$ and any $a(x) \in \mathcal{R}^{np^s}_{k,f_j, k_j}$, define
		\[
		\varphi_j : \mathcal{R}^{np^s}_{k,f_j,k_j} \to \mathcal{A}_j, \quad 
		a(x) \mapsto \varepsilon_j(x)a(x) \mod{x^{np^s}-\lambda}.
		\]
		Then $\varphi_j$ is a ring isomorphism from $\mathcal{R}^{np^s}_{k,f_j, k_j}$ onto $\mathcal{A}_j$.
		
		\item For any $a_j(x) \in \mathcal{R}^{np^s}_{k,f_j, k_j}$, $j = 1, \ldots, t$, define
		\[
		\varphi : \mathcal{R}^{np^s}_{k,f_1, k_1} \times \cdots \times \mathcal{R}^{np^s}_{k,f_t, k_t} \to \mathcal{R}^{np^s}_{k,\lambda}
		\]
		by
		\[
		\varphi(a_1(x), \ldots, a_t(x)) = \sum_{j=1}^t \varepsilon_j(x)a_j(x) \pmod{x^{np^s}-\lambda}.
		\]
		Then $\varphi$ is a ring isomorphism from 
		$\mathcal{R}^{np^s}_{k,f_1, k_1} \times \cdots \times \mathcal{R}^{np^s}_{k,f_t, k_t}$ onto $\mathcal{R}^{np^s}_{k,\lambda}$.
	\end{enumerate}
\end{lemma}
Since $f_j(x)^{k_j}$ and $x^{np^s}-\lambda$ are central, the result follows by standard computations analogous to those in \cite{CCDF18}; hence, we omit the details.

\begin{theorem}\label{decomposition of consta code}
	Let $C \subseteq \mathcal{R}^{np^s}_{k,\lambda}$. Then the following are equivalent:
	\begin{enumerate}
		\item $\mathcal{C}$ is a skew $(\lambda,\Theta)$-constacyclic code of length $np^s$ over $R_k$, i.e., a left ideal of $\mathcal{R}^{np^s}_{k,\lambda}$.
		\item There exist unique left ideal  $\mathcal{C}_j$ of $\mathcal{R}^{np^s}_{k,f_j, k_j}$ such that
		\[
		\mathcal{C} = \bigoplus_{j=1}^t \varepsilon_j(x) \mathcal{C}_j.
		\]

\item Under the canonical isomorphism
\[
\mathcal{R}^{np^s}_{k,\lambda}
\cong
\bigoplus_{j=1}^t \mathcal{R}^{np^s}_{k,f_j, k_j},
\]
the code \(\mathcal{C}\) corresponds to
\[
\bigoplus_{j=1}^t \mathcal{C}_j,
\]
where each \(\mathcal{C}_j\) is a skew $( f_j^{k_j},\Theta )$-polycyclic code of length \(l_jk_j\), and \(l_j=\deg(f_j(x))\).
	\end{enumerate}
\end{theorem}

Hence, from \Cref{decomposition of consta code} in order to study skew $(\lambda,\Theta)$-constacyclic code of length $np^s$ over $R_k$, it is sufficient to study 
left ideal $\mathcal{I}$ of $\mathcal{R}_{k,f_j, k_j}^{np^{s}}$. 
 	Let
	$
	\mathcal{R}^{np^s}_{k, f,j}
	=\frac{R_k[x, \Theta]}{\langle (f(x))^{j}\rangle},
	$
	where $f(x)^{j}$ is an irreducible central polynomial dividing $x^{np^s}-\lambda$ such that degree of $f(x)=l$.
	So naturally $\mu_k$ are also working for 	$
	\mathcal{R}^{np^s}_{k, f, j},~j\in \mathbb{N}$.

	\begin{proposition}\label{ideal over F_{p^m}}
		Let
		$
		\mathcal{R}^{np^s}_{1, f, j}
		=\frac{\mathbb{F}_{p^m}[x;\theta]}{\langle f(x)^{j}\rangle},~j \in \mathbb{N},
		$
		where $f(x)$ is an irreducible central polynomial dividing $x^{np^s}-\alpha,~ \alpha \in \mathbb{F}^*.$ 
		Then $	\mathcal{R}^{np^s}_{1,f,j}$ is a principal left ideal ring.
		Moreover, every left ideal of $	\mathcal{R}^{np^s}_{1, f,j}$ is generated by
		$
		a(x)+\langle f(x)^{j}\rangle,
		$
		where $a(x)$ is a  divisor of $f(x)^{j}$ in $\mathbb{F}_{p^m}[x;\theta]$.
	\end{proposition}

	\begin{definition}The set $\mathcal{B}^{np^s}_{ \mu_k{(f^j)}}$ is defined as the set of all monic divisors of the central polynomial $(\mu_k(f(x)^j))$ in $\mathbb{F}_{p^m}[x;\theta]$, where $\mu_k(f(x))$ is an irreducible central polynomial such that  $\mu_k(f(x)^j)$ dividing $x^{np^s}-\alpha$.
		\[ \mathcal{B}^{np^s}_{ \mu_k(f^j)} = \{ g(x) \in \mathbb{F}_{p^m}[x;\theta] \mid g(x) \text{ is a divisor of } \mu_k(f(x)^j) \}. \]
	\end{definition}
	Following  Theorem 3.5 and Proposition 3.6 of  \cite{HS23},  we have following propositions.
\begin{proposition}\label{ideal over R_2} 
	Every left ideal of $\mathcal{R}^{np^s}_{2, f, j}, ~j \in \mathbb{N}$ is of the form
	\[
	\mathcal{I}
	= \mathcal{R}^{np^s}_{2,f, j}( a_1(x)
	+ u r_{1,1}(x))
	+ \mathcal{R}^{np^s}_{2, f, j}\, (u a_2(x)),
	\]
	where $a_1(x)=0$ or $a_1(x)\in \mathcal{B}^{np^s}_{\mu_2(f^j)}$, 
	$r_{1,1}(x)\in \mathcal{R}^{p^s}_{1, \mu_2(f),j}$,
	$a_2(x)=0$ or $a_2(x)\in \mathcal{B}^{np^s}_{\mu_2(f),j}$,
	$a_2(x)\mid_r a_1(x)$ and  $r_{1,1}(x)$ under the above conditions is unique.
Also, if $a_2(x) \neq 0$ then $\deg(r_{1,1}(x))<\deg(a_2(x))$.

\end{proposition}

	\begin{proposition}\label{special 2}
		The  left ideals of the ring $\mathcal{R}^{np^{s}}_{2,f, j}$, can be separated into the following four types:
		\begin{itemize}
			\item[(i)] \emph{ Trivial ideals:}
			$
			\{0\}, \ \mathcal{R}^{np^{s}}_{2,f, j}.
			$
				\item[(ii)] \emph{ Principal left ideals with non-monic polynomial generators:}
			$
			\mathcal{R}^{np^{s}}_{2, f,j} \,( u a_2(x)),
			$
			where $a_2(x) \in \mathcal{B}^{np^s}_{\mu_2(f^j)}$ and $0 \le \deg(a_2(x)) \le lj-1$.
			
			\item[(iii)] \emph{Principal left ideals:}
			$
			\mathcal{R}^{np^{s}}_{2,f,j}( a_1(x) + u r_{1,1}(x)),
			$
			where $a_1(x) \in \mathcal{B}^{np^s}_{\mu_2(f^j)}$, $1 \le \deg(a_1(x)) \le lj-1$, and
			$\deg(r_{1,1}(x)) < \deg(a_1(x))$.
			Moreover, $r_{1,1}(x)$ satisfying the above conditions is unique.
			
			\item[(iv)] \emph{Non-principal left ideals:}
			$
			\mathcal{R}^{np^{s}}_{2,f,j} (a_1(x) + u r_{1,1}(x)) + \mathcal{R}^{np^{s}}_{2,f,j} (u a_2(x)),
			$
			where $a_1(x), a_2(x) \in \mathcal{B}^{np^s}_{\mu_2(f^j)}$, 
			$1 \le \deg(a_1(x)) \le lj-1$, 
			$a_2(x) \mid_{r} a_1(x)$, and
			$\deg(r_{1,1}(x)) < \deg(a_2(x))$.
			Moreover, $r_{1,1}(x)$ satisfying the above conditions is unique.
		\end{itemize}
	\end{proposition}

%
%
%
\section{ Skew $(\lambda,\Theta)$-Constacyclic Codes of Length $np^s$ over $R_k$}
As shown in the preliminaries (see \Cref{decomposition of consta code}), the study of skew $(\lambda,\Theta)$-constacyclic codes of length $np^s$ over $R_k$ reduces to determining the left ideals $\mathcal{I}$ of $\mathcal{R}_{k,f,j}^{np^s}$. In this section, assuming $\gcd(n,p)=1$ and $s\ge 1$, we investigate these ideals and obtain a complete classification in the case $\lambda=\alpha\in\mathbb{F}_{p^m}^*$, leading to a unified description of their algebraic structure. 
We first generalize \cite[Lemma 2.2]{HS23}.
\begin{lemma}\label{auto R_k}
	For $\theta \in \text{Aut}(\mathbb{F}_{p^m})$, $\eta_1 \in \mathbb{F}_{p^m}^*$, and $\eta_i \in 1 +  u^{i-1} \mathbb{F}_{p^m}$, where $2\leq i \leq k-1$, let
	\[
	\Theta_{\theta, \eta_1, \eta_2, \dots, \eta_{k-1}} : R_k \longrightarrow R_k
	\]
	be defined by
	\[
	\sum_{i=0}^{k-1} a_i u^i \longmapsto \sum_{i=0}^{k-1} \theta(a_i) \left( \prod_{j=1}^{k-1} \eta_j \right)^{i} u^i
	\]
	Then, $\text{Aut}(R_k) = \{ \Theta_{\theta, \eta_1, \eta_2, \dots, \eta_{k-1}} : \theta \in \text{Aut}(\mathbb{F}_{p^m}), \eta_1 \in \mathbb{F}_{p^m}^*, \eta_{i} \in 1 + u^{i-1} \mathbb{F}_{p^m} \}$.
\end{lemma}
\begin{proof} Denote  $\Theta_{\theta, \eta_1, \dots, \eta_{k-1}} $ by $ \Theta$ for brevity. We first check that
	$\Theta  \in \text{Aut}(R_k)$. Indeed, $\Theta$ is a ring homomorphism. For,  if $A = \sum_{i=0}^{k-1} a_i u^i$ and $B = \sum_{j=0}^{k-1} b_j u^j$,  then 
		\[
	\Theta ( A + B ) = \Theta \left( \sum_{i=0}^{k-1} (a_i + b_i) u^i \right) = \sum_{i=0}^{k-1} \theta(a_i + b_i) \left( \prod_{j=1}^{k-1} \eta_j \right)^i u^i.
	\]
	Since $\theta \in \text{Aut}(\mathbb{F}_{p^m})$, we have that $ \theta(a_i + b_i) = \theta(a_i) + \theta(b_i)$. Consequently,
	\[
	\Theta(A+B) = \sum_{i=0}^{k-1} \theta(a_i) \left( \prod_{j=1}^{k-1} \eta_j \right)^i u^i + \sum_{i=0}^{k-1} \theta(b_i) \left( \prod_{j=1}^{k-1} \eta_j \right)^i u^i = \Theta(A) + \Theta(B).
	\]
	Also, using Cauchy product we have,
	\[
	\Theta(AB) = \Theta \left( \sum_{i=0}^{k-1} \left( \sum_{j=0}^{i} a_j b_{i-j} \right) u^i \right) = \sum_{i=0}^{k-1} \theta \left( \sum_{j=0}^{i} a_j b_{i-j} \right) \left( \prod_{j=1}^{k-1} \eta_j \right)^i u^i
	\]
	\[
	= \sum_{i=0}^{k-1} \left( \sum_{j=	0}^{i} \theta(a_j) \theta(b_{i-j}) \right) \left( \prod_{j=1}^{k-1} \eta_j \right)^i u^i = \Theta(A)\Theta(B).
	\]
	To check injectivity of $\Theta$, 	suppose $\Theta \left( \sum_{i=0}^{k-1} a_i u^i \right) = 0$, which implies
	\[
	\theta(a_0) + \theta(a_1)\left( \prod_{j=1}^{k-1} \eta_j \right)u + \theta(a_2)\left( \prod_{j=1}^{k-1} \eta_j \right)^2 u^2 + \dots +\theta(a_{k-1})\left( \prod_{j=1}^{k-1} \eta_j \right) ^{k-1}u^{k-1}= 0.
		\] 
		By comparing the coefficient of $1$ (the constant term) on both sides, we immediately obtain 
	\[
\prod_{j=1}^{k-1} \eta_j \left(\theta(a_1)u + \theta(a_2)\left( \prod_{j=1}^{k-1} \eta_j \right) u^2 + \dots +\theta(a_{k-1})\left( \prod_{j=1}^{k-1} \eta_j \right) ^{k-2}u^{k-1}\right)= 0.
		\] 
		Since $\prod_{j=1}^{k-1} \eta_j $ is invertible, hence we get 
			\[
		\theta(a_1)u + \theta(a_2)\left( \prod_{j=1}^{k-1} \eta_j \right) u^2 + \dots +\theta(a_{k-1})\left( \prod_{j=1}^{k-1} \eta_j \right) ^{k-2}u^{k-1}= 0.
		\] 
		Now, by comparing the coefficient of $u$ on both sides, we get that $a_1=0$.  
	Continuing this recursive process for higher powers of $u$ yields
	$
	\theta(a_0) = 0, \theta(a_1) = 0, \dots, \theta(a_{k-1}) = 0.
	$
	 Thus, $\text{ker} ~\Theta= \{0\}$ i.e., $\Theta$ is injective. As  $R_k$ is a finite ring, we have that $\Theta \in \text{Aut}(R_k)$.\\
	  We now check  the set $ \{ \Theta_{\theta, \eta_1, \eta_2, \dots, \eta_{k-1}} : \theta \in \text{Aut}(\mathbb{F}_{p^m}), \eta_1 \in \mathbb{F}_{p^m}^*, \eta_i \in 1 + \sum_{j=1}^{i} u^j \mathbb{F}_{p^m} \}$ is a complete set of automorphisms of $R_k$.
	 Let $\Phi \in \text{Aut}(R_k)$. Since $\mathbb{F}_{p^m}$ is the residue field of $R_k$, the restriction of $\Phi$ to $\mathbb{F}_{p^m}$ gives a field automorphism $\theta \in \text{Aut}(\mathbb{F}_{p^m})$.
	 		As $\langle u \rangle$ is the unique maximal ideal of $R_k$, it must map to itself under any ring automorphism of $R_k$  and
	 		 hence,	 $\Phi(u) = \eta u,$
	 	for some unit $\eta \in R_k^*$.	 Any  element of $R_k^*$ can be written as	$
	 	\eta = \eta_1 \eta_2 \cdots \eta_{k},$
	 	where $\eta_1 \in \mathbb{F}_{p^m}^*$ and $\eta_i \in 1 +  u^{i-1} \mathbb{F}_{p^m}$ for $2 \leq i \leq k$. Such a decomposition exists because every unit in $R_k$ admits a unique expansion of this form.
	 	Now, $	\Phi(u)=( \eta_1 \eta_2 \cdots \eta_{k})u= (\eta_1 \eta_2 \cdots \eta_{k-1})u.$ 
	 	Consequently, $\Phi(u^i)=\Phi(u)^i = ((\eta_1 \eta_2 \cdots \eta_{k-1}) u)^i = (\eta_1 \eta_2 \cdots \eta_{k-1})^i u^i$, and we obtain
	 	\[
	 	\Phi\left(\sum_{i=0}^{k-1} a_i u^i\right)
	 	= \sum_{i=0}^{k-1} \theta(a_i)\left(\prod_{j=1}^{k-1}\eta_j\right)^i u^i.
	 	\]
	 		Thus $\Phi = \Theta_{\theta,\eta_1,\dots,\eta_{k-1}}$.
\end{proof}

\begin{theorem}\label{ideals R_k} Every skew $(f^j, \Theta)$-polycyclic code of length $lj$ where $l=$ $\deg f(x)$, i.e.,
	every left ideal $\mathcal{I}$ of $\mathcal{R}_{k,f,j}^{np^{s}}$ is of the form
	\[ \mathcal{I} = \sum_{i=1}^{k} \mathcal{R}_{k,f, j}^{np^{s}} \left(u^{i-1}a_{i}(x) + \sum_{j=1}^{k-i} u^{i-1+j}r_{i,j}(x)\right), \]
	satisfying the following conditions:
	$a_{i}(x)=0 ~\text{or}~ a_{i}(x) \in \mathcal{B}^{np^s}_{\mu_k(f^j)}$, $a_{k}(x) ~|_{r} ~a_{i}(x)$ for all $1 \le i \le k-1$,   $r_{i,j}(x)\in \mathcal{R}^{np^s}_{1, \mu_k(f),j}.$
	Moreover,  the polynomials $r_{i,j}(x)$ satisfying these conditions are unique.
	Also, if $a_{i}(x) \neq 0$ then for each $2 \le i \le k$, we have,
	\[ \max_{1 \le j \le i-1} \{\deg r_{j,i-j}(x)\} < \deg a_{i}(x).\]

\end{theorem}
\begin{proof}
	We will prove the theorem by induction on $k$. From  \Cref{ideal over F_{p^m}} and \Cref{ideal over R_2} we have the  base cases for $k=1, 2$.  Now assume that the theorem holds for $k-1$. Let $\mathcal{I}$ be a left ideal of $\mathcal{R}_{k,f,j}^{np^s}$ is of the form $\mathcal{I} = \frac{I }{ \langle f(x)^{j} \rangle}$, where $I$ is a left ideal in $R_k[x;\Theta]$ containing the central ideal $\langle f(x)^{j} \rangle$. Since $\mu_k : R_k[x;\Theta] \longrightarrow \mathbb{F}_{p^m}[x;\theta]$ is an epimorphism, $\mu_k(I)$ is a left ideal of $\mathbb{F}_{p^m}[x;\theta]$. Thus, there exists a polynomial $h(x) \in R_k[x;\Theta]$ such that $\mu_k(I)=\mathbb{F}_{p^m}[x;\theta]\mu_k(h(x))$. 
	It is easy to see that that $R_k[x;\Theta](h(x)) + (I \cap R_k[x;\Theta]u) \subseteq I.$ 
	Let $h_1(x)$ be an arbitrary polynomial in $I$. Then $\mu_k(h_1(x)) \in \mu_k(I)=\mathbb{F}_{p^m}[x;\theta]\mu_k(h(x))$ and hence there exists $t(x) \in R_k[x;\Theta]$ such that $\mu_k(h_1(x)) = \mu_k(t(x)) \mu_k(h(x))$ which implies that $h_1(x) - t(x)h(x) \in \ker(\mu_k) \cap I = I \cap R_k[x;\Theta]u$. Thus, $h_1(x) \in R_k[x;\Theta]h(x) + (I \cap R_k[x;\Theta]u)$ and therefore, $I = R_k[x;\Theta]h(x) + (I \cap R_k[x;\Theta]u)$. Since $\langle f(x)^{j}\rangle \subseteq$ $I$, it implies that 
	\[ \mathcal{I} = \mathcal{R}_{k,f,j}^{np^{s}} (h(x)) +u\mathcal{R}_{k,f,j}^{np^s} \cap \mathcal{I} . \]
Furthermore, we observe that $u\mathcal{R}_{k,f,j}^{np^s} \cap \mathcal{I} = u \pi_{k}^{-1}(\pi_{k}(\mathcal{I} : u))$, where $\pi_{k}(\mathcal{I} : u)$ is an ideal of $\mathcal{R}_{k-1,\pi_k(f),j}^{np^{s}}$. By the induction hypothesis, left ideals of $\mathcal{R}_{k-1,\pi_k(f),j}^{np^{s}}$ are of the form
	\[ \mathcal{I}_{k-1} = \sum_{i=1}^{k-1} \mathcal{R}_{k-1,\pi_k(f),j}^{np^{s}} \left(u^{i-1}a_{i+1}(x) + \sum_{j=1}^{k-1-i} u^{i-1+j}r_{i+1,j}(x)\right)\]
	satisfying the following conditions that
	$a_{i+1}(x) \in  \mathcal{B}^{np^s}_{\mu_{k-1}(f^j)}$, $a_{k}(x) ~|_{r} ~a_{i+1}(x)$ for all $1 \le i \le k-2$, and for each $2 \le i \le k-1$,
	$\max_{1 \le j \le i-1} \{\deg r_{j,i-j}(x)\} < \deg a_{i+1}(x).$
	Also,  $r_{i+1,j}(x)\in \mathcal{R}^{np^s}_{1, \mu_k(f),j}$ with above conditions are unique.
	The pre-image $\pi_k^{-1}$ pulls this structure from $\mathcal{R}^{np^s}_{k-1, \pi_k(f),j}$ back to $\mathcal{R}^{np^s}_{k, f,j}$ then after multiplying by $u$ we have,
	
	\[ u \mathcal{R}_{k,f,j}^{np^s} \cap \mathcal{I} = \sum_{i=1}^{k-1} \mathcal{R}_{k,f,j}^{np^{s}} \left(u^{i}a_{i+1}(x) + \sum_{j=1}^{k-1-i} u^{i+j}r_{i+1,j}(x)\right), \] where 	$a_{i+1}(x) \in  \mathcal{B}^{np^s}_{\mu_{k}(f^j)}$.
	By \Cref{ideal over F_{p^m}} we have
	\[
	\mu_k(\mathcal{I})=\mathcal{R}_{1,\mu_k(f),j}^{np^s}(\mu_k(h(x))=\mathcal{R}_{1,\mu_k(f),j}^{np^s}(a_1(x)),
	\]
	where $a_{1}(x)=0$ or $a_{1}(x) \in  \mathcal{B}^{np^s}_{\mu_k(f^j)}.$ Hence, there exist an element $v(x)$ such that $a_{1}(x)=v(x)\mu_k(h).$ Therefore, $\mu_k(h(x))$ is factor of $(\mu_k(f(x)^j)).$ We may assume, without loss of generality, $h(x)=a_1(x)+ug_{1,1}(x)+u^2g_{1,2}(x)+\cdots+u^{k-1}g_{1,k-1}(x)$ where $g_{1,j}(x)\in \mathcal{R}^{np^s}_{1, \mu_k(f),j}.$
	If $a_1(x)=0$ then $u(g_{1,1}(x)+ug_{1,2}(x)+\cdots+u^{k-2}g_{1,k-1}(x)) \in u\mathcal{R}_{k,f,j}^{np^s} \cap \mathcal{I}$, hence 
		\[
	\mathcal{I}=\sum_{i=1}^{k-1}\mathcal{R}_{k,f,j}^{np^s}\left(u^{i}a_{i+1}(x)+\sum_{j=1}^{k-i-1}u^{i+j}r_{i+1,j}(x)\right).
	\]
	 Now if $a_1(x) \neq 0$ then
	\[
	\mathcal{I}=\mathcal{R}_{k,f,j}^{np^s}\left(a_1(x)+\sum_{i^{'}=1}^{k-1}u^{i}g_{1,i{'}}(x)\right)+\sum_{i=1}^{k-1}\mathcal{R}_{k,f,j}^{np^s}\left(u^{i}a_{i+1}(x)+\sum_{j=1}^{k-i-1}u^{i+j}r_{i+1,j}(x)\right).
	\]
	If $a_1(x)\mid_r a_k(x)$, then  by induction hypothesis, we have $a_1(x)\mid_r a_2(x)$, which contradicts \Cref{ideal over R_2}. Therefore, it follows that $a_k(x)\mid_r a_1(x)$.
 Hence, we have $a_k(x) ~|_r~a_i(x)$ for all $1\leq i \leq k-1.$ 
	Suppose $a_{i+1}(x)\neq 0,$ apply right division algorithm on $g_{1,i}$ and $a_{i+1}$, we get
	\[
	g_{1,i}(x)=q_{1,i}(x)a_{i+1}(x)+r_{1,i}(x),
	\]
	where \begin{equation} \label{eq 1}
		r_{1,i}(x)=0~ \text{or deg}(r_{1,i}(x))< \text{deg} (a_{i+1}(x)).
	\end{equation}
	Therefore,
	\[
	a_1(x)+\sum_{i=1}^{k-1}u^{i}r_{1,i}(x)=a_1(x)+\sum_{i=1}^{k-1}u^{i}g_{1,i}-\sum_{i=1}^{k-1}u^{i}q_{1,i}(x)a_{i+1}(x) \in \mathcal{I}
	\]
	Thus, we can write $$\mathcal{I}=\mathcal{R}_{k,f,j}^{np^s}\left(a_1(x)+\sum_{i^{'}=1}^{k-1}u^{i}r_{1,i{'}}(x)\right)+\sum_{i=1}^{k-1}\mathcal{R}_{k,f,j}^{np^s}\left(u^{i}a_{i+1}(x)+\sum_{j=1}^{k-i-1}u^{i+j}r_{i+1,j}(x)\right),$$ where  for each $2 \le i \le k$,
	\[ \max_{1 \le j \le i-1} \{\deg r_{j,i-j}(x)\} < \deg a_{i}(x). \]
	For the uniqueness of $r_{1,i}$, for $1 \leq i \leq k-1$, suppose that $r'_{1,i}$ is another polynomial, such that 
	$\mathcal{I}=\mathcal{R}_{k,f,j}^{np^s}\left(a_1(x)+\sum_{i=1}^{k-1}u^{i}r'_{1,i}(x)\right)+\sum_{i=1}^{k-1}\mathcal{R}_{k,f,j}^{np^s}\left(u^{i}a_{i+1}(x)+\sum_{j=1}^{k-i-1}u^{i+j}g_{i+1,j}(x)\right)$.
	Hence
	$(a_1(x)+\sum_{i=1}^{k-1}u^ir_{1,i})- (a_1(x)+\sum_{i=1}^{k-1}u^ir'_{1,i}) \in \mathcal{I}$. which implies that  $\sum_{i=1}^{k-1}u^i(r_{1,i}-r'_{1,i}) \in \mathcal{I}.$ This implies $r_{1,i}-r'_{1,i} 
	\in Tor_{i+1} (\mathcal{I})=\mathcal{R}_{1, \mu_k(f),j}^{np^s}(a_{i+1}(x)).$ If $r_{1,i} \neq r'_{1,i},$ then  deg $a_{i+1}(x) \leq$ deg $(r_{1,i}-r'_{1,i})$, which is a contradiction to \Cref{eq 1}. Thus, $r_{1,i}=r'_{1,i} ~\forall i.$ 
\end{proof}

In particular, if $\lambda=\alpha$, where $\alpha \in \mathbb{F}_{p^m}^*$, then we classify all left ideals of  $ \mathcal{R}_{k,f,j}^{np^s}$. 
Thus, we are able to classify all skew cyclic and skew negacyclic codes of length $n p^s$ over 
$\mathcal{R}^{n p^{s}}_{k, f,j} $ for $\lambda=\pm 1$. 
\begin{theorem}\label{alpha ideals}
		Every left ideals of $\mathcal{R}^{np^{s}}_{k,f,j}$ can be classified as
	\[ \mathcal{I} = \sum_{i=1}^{k} \alpha_i~ \mathcal{R}_{k,f,j}^{np^{s}}\left(u^{i-1}a_{i}(x) + \sum_{j=1}^{k-i} u^{i-1+j}r_{i,j}(x)\right),  \]	where $\alpha_i\in\{0,1\}$ with $(\alpha_1,\ldots,\alpha_k)\neq(0,\ldots,0)$,
	satisfying the following conditions
	$a_{i}(x)=0~ \text{or}~ a_{i}(x) \in \mathcal{B}^{np^s}_{\mu_k(f^j)}$, $a_{k}(x) ~|_{r} ~a_{i}(x)$ for all $1 \le i \le k-1$, and  $r_{i,j}(x)\in \mathcal{R}^{np^s}_{1, \mu_k(f),j}$. 	Also, the polynomials $r_{i,j}(x)$ satisfying these conditions are unique.
	 If $a_{i}(x) \neq 0$ then for each $2 \le i \le k$, we have,
	\[ \max_{1 \le j \le i-1} \{\deg r_{j,i-j}(x)\} < \deg a_{i}(x).\]

\end{theorem}

\begin{proof}
	We prove the theorem by induction on $k$.	For $k=1$ and $k=2$, the result follows from \Cref{ideal over F_{p^m}}
	and \Cref*{special 2}, respectively. Assume that the theorem holds for $k-1$.
	Let $\mathcal{I}$ be a left ideal of $\mathcal{R}^{np^s}_{k,f,j}$.
	As shown earlier,
	\[
	\mathcal{I}
	= \mathcal{R}^{np^s}_{k,f,j} (h(x))
	+ u\,\pi_k^{-1}\bigl(\pi_k(\mathcal{I}:u)\bigr),
	\]
	where $\pi_k(\mathcal{I}:u)$ is a left ideal of
	$\mathcal{R}^{np^s}_{k-1,\pi_k(f),j}$ and  $h(x) \in R_k[x;\theta]$ such that $	\mu_k(\mathcal{I})=\mathcal{R}_{1,\mu_k(f),j}^{np^s}(\mu_k(h(x))$.
	Case 1: Assume $	\mathcal{I} \subseteq   \mathcal{R}^{np^s}_{k,f,j}u .$ Then $	\mathcal{I}= u\,\pi_k^{-1}\bigl(\pi_k(\mathcal{I}:u)\bigr).$
	By the induction hypothesis,
	
	\[ \pi_k(\mathcal{I}:u) = \sum_{i=1}^{k-1}\alpha_i \mathcal{R}_{k-1,\pi_k(f),j}^{np^{s}} \left(u^{i-1}a_{i+1}(x) + \sum_{j=1}^{k-1-i} u^{i-1+j}r_{i+1,j}(x)\right)\]

		satisfying the following conditions
	\begin{equation}\label{eq 2}
		\begin{aligned}
			& \alpha_i \in \{0,1\} \quad \text{and} \quad (\alpha_1, \alpha_2, \ldots, \alpha_{k-1}) \neq (0,0,\ldots,0),\\
			& a_{i+1}(x) \in \mathcal{B}^{np^s}_{\mu_{k-1}(f^j)}, ~
			a_{k}(x) \mid_r a_{i+1}(x), \quad \text{for all } 1 \le i \le k-2. \\
			&\text{For each } 2 \le i \le k-1, \quad  \text{if} \quad a_{i+1} \neq 0\quad \text{then} \quad
			\max_{1 \le j \le i-1} \{ \deg r_{j,\,i-j}(x) \} < \deg a_{i+1}(x).
				\end{aligned}				
	\end{equation}
		Also, the polynomials  $r_{i+1,j}(x)\in \mathcal{R}^{np^s}_{1, \mu_k(f),j}$ satisfying these conditions are unique.
		Thus,
		\[ 	\mathcal{I}= u\,\pi_k^{-1}\bigl(\pi_k(\mathcal{I}:u)\bigr) = \sum_{i=1}^{k-1} \alpha_i \mathcal{R}_{k,f,j}^{np^{s}}\left(u^{i}a_{i+1}(x) + \sum_{j=1}^{k-i} u^{i+j}r_{i+1,j}(x)\right),  \]
which satisfies all the required conditions.\\
	Case 2:  Assume $\mathcal{I} \nsubseteq  \mathcal{R}^{np^s}_{k,f,j} u .$ 
	So, 	\[
	\mathcal{I}
	= \mathcal{R}^{np^s}_{k,f,j} (h(x))
	+ u\,\pi_k^{-1}\bigl(\pi_k(\mathcal{I}:u)\bigr),
	\]
	where $\pi_k(\mathcal{I}:u)$ is a left ideal of
	$\mathcal{R}^{np^s}_{k-1,\pi_k(f),j}$ and  $h(x) \in R_k[x;\Theta]$ such that $	\mu_k(\mathcal{I})=\mathcal{R}_{1,\mu_k(f),j}^{np^s}(\mu_k(h(x))$ and $\mu_k(h(x)) \neq 0$.
	So by similar argument as we use in above theorem we can write, 
	$$ \mathcal{R}^{np^s}_{k,f,j} (h(x))=\mathcal{R}_{k,f,j}^{np^s}\left(a_1(x)+\sum_{i^{'}=1}^{k-1}u^{i}r_{1,i{'}}(x)\right),$$	where \begin{equation} \label{eq 3}
		r_{1,i}(x)=0~ \text{or deg}(r_{1,i}(x))< \text{deg} (a_{i+1}(x)), ~ a_k(x) ~|_r~a_1(x).
	\end{equation}
	Also  the polynomials $r_{1,j}(x)$ satisfying these conditions are unique.
	So	\[
	\mathcal{I}=\mathcal{R}_{k,f,j}^{np^s}\left(a_1(x)+\sum_{i^{'}=1}^{k-1}u^{i}g_{1,i{'}}(x)\right)+ \sum_{i=1}^{k-1} \alpha_i \mathcal{R}_{k,f,j}^{np^{s}}\left(u^{i}a_{i+1}(x) + \sum_{j=1}^{k-i} u^{i+j}r_{i+1,j}(x)\right),  
	\]
	By combining conditions of  \ref{eq 2}, \ref{eq 3} we get
	$a_{i}(x) \in \mathcal{B}^{np^s}_{\mu_k{(f^j)}}$, $a_{k}(x)~ |_{r}~ a_{i}(x)$ for all $1 \le i \le k-1$. If $a_{i+1} \neq 0$ then for each $2 \le i \le k$, we have 
	\[ \max_{1 \le j \le i-1} \{\deg r_{j,i-j}(x)\} < \deg a_{i}(x), \] where $\alpha_i=0 ~\text{or} ~1$ \text{and}~ $(\alpha_1, \alpha_2, \ldots, \alpha_k) \neq (0, 0, \ldots, 0)$. The polynomials $r_{i,j}(x)$ satisfying these conditions are unique.
\end{proof}
\begin{remark}
	In particular, when $\Theta$ is the identity automorphism and $\lambda=1$, Theorems~3.2 and~3.3 of \cite{HMRS24} follow as corollaries of the results obtained in this section.
\end{remark}

We now establish an isomorphism between skew $\Theta$-cyclic codes of length $np^s$ and skew $(\lambda, \Theta)$- constacyclic codes over $R_k$.\\
For $\alpha \in \mathbb{F}_{p^m}$, we have that $\alpha^{p^m}=\alpha$. So inductively,  for any positive integer $k$, we have $\alpha^{({p^m})^k}=\alpha$.
By the division algorithm, there exist non-negative integers $q$ and $r $ such that
\[
s=q m+r,\quad 0\le r \le m-1.
\]
Let
\[
\alpha_0'=\alpha^{-p^{m-r}}.
\]
Then
\[
\alpha_0'^{p^s}
=\alpha^{-p^{(q+1)m}}
=\alpha^{-1}.
\]

\begin{proposition}\label{lemma_uniform_iso}
	Suppose $\lambda = \alpha^n, \alpha \in \mathbb{F}_{p^m}^*, n\in \mathbb{N}$. If $\Theta(\alpha) = \alpha$, then the map 
	\[ \Psi: \frac{R_k[x; \Theta]}{\langle x^{np^s} - 1 \rangle} \to \frac{R_k[x; \Theta]}{\langle x^{np^s} - \lambda \rangle} \]
	defined by $\Psi(f(x)) = f(\alpha_0' x)$ is a weight-preserving ring isomorphism. In particular,  $\Psi$ provides a one-to-one correspondence between  skew $\Theta$-cyclic codes and skew $(\lambda, \Theta)$-constacyclic code of length $np^s$.
\end{proposition}
\begin{proof}
	To prove $\Psi$ is a well-defined,
consider $f(x) \equiv g(x) \pmod{x^{np^s} - 1}$. This implies $f(x) - g(x) = h(x)(x^{np^s} - 1)$ for some $h(x) \in R_k[x; \Theta]$. Now replace $x$ by $\alpha_0'x$ both side,
	\begin{equation}
	f(\alpha_0' x) - g(\alpha_0' x) = h(\alpha_0 'x) \left( (\alpha_0' x)^{np^s} - 1 \right).
	\end{equation} 
Because $\Theta(\alpha) = \alpha$, the automorphism $\Theta$ fixes the constants $\alpha_0'$.
It follows that $(\alpha_0' x)^{np^s} = \alpha_0'^{np^s} x^{np^s}$. As $\alpha_0'^{p^s} = \alpha^{-1}$, we have,
	$ \alpha_0'^{np^s} = \lambda^{-1}. $
	\\Hence,
	\[ f(\alpha_0' x) - g(\alpha_0' x) = h(\alpha_0' x)(\lambda^{-1}x^{np^s} - 1) = \lambda^{-1}h(\alpha_0' x)(x^{np^s} - \lambda). \]
	Thus, $f(\alpha_0' x) \equiv g(\alpha_0' x) \pmod{x^{np^s}- \lambda}$, confirming that $\Psi$ is a well-defined. 
Also, clearly $\Psi$ is a ring homomorphism.
	Since $\alpha_0'$ is a unit in $\mathbb{F}_{p^m}$, it is also a unit in $R_k$. For an arbitrary element $f(x) \in \frac{R_k[x; \Theta]}{\langle x^{np^s} - \lambda \rangle}$, there exist an element  $f(\alpha_0x) \in  \frac{R_k[x; \Theta]}{\langle x^{np^s} - 1 \rangle}$ such that  $\Psi^{-1}(f(x)) = f(\alpha_0x),$ where  $\alpha_0'^{-1}=\alpha_0$. Clearly, $\Psi$ is onto, hence $\Psi$ is an isomorphism.
		Finally, the map is weight-preserving because $\Psi(\sum a_i x^i ) = \sum (a_i \alpha_0'^i) x^i , a_i \in R_k$. Since $\alpha_0'$ is a unit, the coefficient $a_i \alpha_0'^i$ is non-zero if and only if $a_i$ is non-zero, preserving the Hamming weight. 
\end{proof}
	We now demonstrate the results developed so far by considering the particular code lengths $3p^s$ and $6p^s$.

	\section{Skew constacyclic codes of length $3p^s$ over $R_k$}
First we will study the skew constacyclic codes of length $3p^s$ over $\mathbb{F}_{p^m}$.
	\subsection{Skew constacyclic codes of length $3p^s$ over $\mathbb{F}_{p^m}$}
	Let \( p \) be an odd prime with \( \gcd(p,3)=1 \), and let \( \theta \) be an automorphism of \( \mathbb{F}_{p^m} \) such that the order of \( \theta \), denoted by \( |\theta| \), divides \( p^s \).
	As observed in Section~3, for any \( \lambda \in \mathbb{F}_{p^m} \), there exists \( \lambda_0 \in \mathbb{F}_{p^m} \) such that
	$\lambda_0^{p^s}=\lambda.$
	Moreover, if \( \theta(\lambda)=\lambda \), then \( \theta(\lambda_0)=\lambda_0 \).
	Hence, $x^{3p^s}-\lambda=(x^3-\lambda_0)^{p^s}.$
	If \( \lambda_0 \) is not a cube in \( \mathbb{F}_{p^m} \), then \( x^3-\lambda_0 \) is irreducible in \( \mathbb{F}_{p^m}[x;\theta] \), and hence the structure of left ideals follows from \Cref{ideals R_k}.
	Now suppose that \( \lambda_0 \) is a cube in \( \mathbb{F}_{p^m} \), and we now undertake this case.\\
	We begin by proving the following lemma, which shall be crucial for further results.
	\begin{lemma}\label{Pro 3p^s iso}
	Let $\lambda \in \mathbb{F}_{p^m}^*$ be  such that $\lambda = \alpha^{3}$ for some $\alpha \in \mathbb{F}_{p^m}^*$ and let $\theta$ be automorphism of $\mathbb{F}_{p^m}$  such that $\theta(\lambda)=\lambda,$  the order  of $\theta$ divides $p^s$. The map 
	\[ \Phi: \frac{\mathbb{F}_{p^m}[x; \theta]}{\langle x^{3p^s} - 1 \rangle} \longrightarrow \frac{\mathbb{F}_{p^m}[x; \theta]}{\langle x^{3p^s} - \lambda \rangle} \]
	defined by $\Phi(f(x)) = f(\alpha_0' x)$, where $\alpha_0'^{p^s} = \alpha^{-1}$, is a weight-preserving ring isomorphism.
\end{lemma}
\begin{proof}
First, we show that \( \theta \) fixes \( \alpha \). Since \( \theta(\lambda)=\lambda \) and \( \lambda=\alpha^3 \), we have \( \theta(\alpha)^3=\alpha^3 \). If \( \mathbb{F}_{p^m}^* \) does not contain a primitive third root of unity, then \( \theta(\alpha)=\alpha \).  In case, \( \mathbb{F}_{p^m}^* \) contains a primitive third root of unity \( \omega \), then \( \theta(\alpha)\in\{\alpha,\alpha\omega,\alpha\omega^2\} \).
 If \( \theta(\alpha)=\alpha\omega \), then \( \theta^2(\alpha)=\theta(\alpha\omega)=\theta(\alpha)\theta(\omega)=\alpha\omega\,\theta(\omega) \). Since \( \theta(\omega)\in\{\omega,\omega^2\} \), if \( \theta(\omega)=\omega \), then \( \theta^2(\alpha)=\alpha\omega^2 \) and \( \theta^3(\alpha)=\alpha \). Let \( |\theta|=p^i \), \( 1\le i\le s \). As \( \gcd(p,3)=1 \), we have \( p^i\equiv 1 \) or \( 2 \mod{3} \), hence \( \theta^{p^i}(\alpha)=\theta^{\pm1}(\alpha)\neq\alpha \), a contradiction. If \( \theta(\omega)=\omega^2 \), then \( \theta^2(\alpha)=\alpha\omega^3=\alpha \), implying \( |\theta| \) is even, again a contradiction since \( |\theta|\mid p^s \) with \( p \) odd. Similarly the case \( \theta(\alpha)=\alpha\omega^2 \) gives a contradiction.  Therefore, \( \theta(\alpha)=\alpha \).
Hence, applying \Cref{lemma_uniform_iso} with length \( 3p^s \) and \( \alpha_0'^{p^s}=\alpha^{-1} \), we conclude that \( \Phi \) is a weight-preserving ring isomorphism.
\end{proof}
 By  \Cref{Pro 3p^s iso}, the study of skew $(\lambda, \theta)$-constacyclic codes of length $3p^s$ is mathematically equivalent to the study of skew  $\theta$-cyclic codes of the same length through the isomorphism $\Phi$. Consequently, our investigation reduces to determining the algebraic structure of the ambient ring,
\[ \mathcal{R}^{3p^s}_{1} = \frac{\mathbb{F}_{p^m}[x; \theta]}{\langle x^{3p^s} - 1 \rangle}. \]
	\begin{theorem}\label{3p^s first thm}
 Suppose that $p \equiv 1 \mod{3}$. 
	Let $\theta \in \mathrm{Aut}(\mathbb{F}_{p^m})$ such that $|\theta|\mid p^s$.  Then every skew $\theta$-cyclic code $\mathcal{C}$ of length $3p^s$ over $\mathbb{F}_{p^m}$ admits a decomposition
	$
	\mathcal{C} = \mathcal{C}_1 \oplus \mathcal{C}_2 \oplus \mathcal{C}_3,
	$
	where
$\mathcal{C}_i$	is left ideal of $\frac{\mathbb{F}_{p^m}[x;\theta]}{\langle  x^{p^s}-\omega^{i-1}\rangle} , ~ i=1,2,3$, respectively and $\omega \in \mathbb{F}_p^*$ is a primitive cube root of unity. Moreover, $|\mathcal{C}| = |\mathcal{C}_1|\,|\mathcal{C}_2|\,|\mathcal{C}_3|.$  Furthermore,
	$\mathcal{C}^\perp = \mathcal{C}_1^\perp \oplus \mathcal{C}_2^\perp \oplus \mathcal{C}_3^\perp,
	$	where $\mathcal{C}_i^\perp$ is a left ideal of the quotient ring $\frac{\mathbb{F}_{p^m}[x; \theta]}{\langle x^{p^s}  - \omega^{1-i} \rangle}$ and  $\mathcal{C}$ is a self-dual code if and only if $\mathcal{C}_1 = \mathcal{C}_1^\perp$, $\mathcal{C}_2 = \mathcal{C}_3^\perp$, and $\mathcal{C}_3 = \mathcal{C}_2^\perp$.
\end{theorem}

\begin{proof}
Any skew $\theta$-cyclic code $\mathcal{C}$ of length $3p^s$ corresponds to a left ideal of
	$
	\mathcal{R}_1^{3p^s} = \frac{\mathbb{F}_{p^m}[x;\theta]}{\langle x^{3p^s}-1 \rangle}.
	$
		Since  the order of $\theta$ divides $p^s$, the polynomial $x^{3p^s}-1$ lies in the center of $\mathbb{F}_{p^m}[x;\theta]$. Thus,
	\[
	x^{3p^s}-1 = (x^{p^s})^3 - 1 = (x^{p^s}-1)(x^{2p^s}+x^{p^s}+1).
	\]
	
	As $p \equiv 1 \mod{3}$, there exists a primitive cube root of unity $\omega \in \mathbb{F}_p^*$ satisfying
	$	\omega^2 + \omega + 1 = 0.$
	Since $\theta$ fixes $\mathbb{F}_p$, we have $\theta(\omega)=\omega$ and $\theta(\omega^2)=\omega^2$. Hence,	$
	x^{2p^s}+x^{p^s}+1 = (x^{p^s}-\omega)(x^{p^s}-\omega^2),
	$
	and therefore
	$	x^{3p^s}-1 = (x^{p^s}-1)(x^{p^s}-\omega)(x^{p^s}-\omega^2).
	$ As $|\theta|\mid p^s$, by \Cref*{Central poly},
	each of the ideals $\langle x^{p^s}-1 \rangle$, $\langle x^{p^s}-\omega \rangle$, and $\langle x^{p^s}-\omega^2 \rangle$ is a two-sided ideal in $\mathbb{F}_{p^m}[x;\theta]$, and these factors are pairwise coprime. Hence, by the Chinese Remainder Theorem, we have
	\[
	\mathcal{R}_1^{3p^s}
	\cong
	\frac{\mathbb{F}_{p^m}[x;\theta]}{\langle x^{p^s}-1\rangle}
	\oplus
	\frac{\mathbb{F}_{p^m}[x;\theta]}{\langle x^{p^s}-\omega\rangle}
	\oplus
	\frac{\mathbb{F}_{p^m}[x;\theta]}{\langle x^{p^s}-\omega^2\rangle}.
	\]
	Thus any left ideal $\mathcal{C} \subseteq \mathcal{R}_1^{3p^s}$ decomposes uniquely as
	$
	\mathcal{C} = \mathcal{C}_1 \oplus \mathcal{C}_2\oplus \mathcal{C}_3,
	$
	where each $\mathcal{C}_i$ is a left ideal of the corresponding component ring.
	\end{proof}
The following theorem establishes the algebraic structure of the dual of the skew $(\lambda, \theta)$-constacyclic code $\mathcal{C}$ of length $3p^s$ over $\mathbb{F}_{p^m}$.

%
%
%
%
\begin{theorem}\label{3p^s 2nd theorem F}
	Let $p \equiv 2 \mod{3}$ and let $\theta \in \mathrm{Aut}(\mathbb{F}_{p^m})$ with $|\theta|\mid p^s$. 
	Let $\mathcal{C}$ be a skew $\theta$-cyclic code of length $3p^s$ over $\mathbb{F}_{p^m}$.
		\begin{enumerate}
			\item[(a)] If $m$ is odd, then $\mathcal{C}$ decomposes as
			$\mathcal{C} = \mathcal{C}_1 \oplus \mathcal{C}_2,
			$
			where $\mathcal{C}_1$ and $\mathcal{C}_2$
			are left ideals of $\frac{\mathbb{F}_{p^m}[x;\theta]}{\langle  x^{p^s}-1\rangle}$ and $\frac{\mathbb{F}_{p^m}[x;\theta]}{\langle  x^{2p^s}+ x^{p^s}+1\rangle}$, respectively. Moreover,
			$	|\mathcal{C}| = |\mathcal{C}_1|\,|\mathcal{C}_2|, \quad 
			\mathcal{C}^\perp = \mathcal{C}_1^\perp \oplus \mathcal{C}_2^\perp,
			$ where $\mathcal{C}_1^\perp $ and $\mathcal{C}_2^\perp $ are left ideals of $\frac{\mathbb{F}_{p^m}[x;\theta]}{\langle  x^{p^s}-1\rangle}$ and $\frac{\mathbb{F}_{p^m}[x;\theta]}{\langle  x^{2p^s}+ x^{p^s}+1\rangle}$, respectively
			and $\mathcal{C}$ is self-dual if and only if $\mathcal{C}_1 = \mathcal{C}_1^\perp$ and $\mathcal{C}_2 = \mathcal{C}_2^\perp$.
			
			\item[(b)] If $m$ is even, then $\mathcal{C}$ decomposes as
			$
			\mathcal{C} = \mathcal{C}_1 \oplus \mathcal{C}_2 \oplus \mathcal{C}_3,
			$
			where $\mathcal{C}_i$	is left ideal of $\frac{\mathbb{F}_{p^m}[x;\theta]}{\langle  x^{p^s}-\omega^{i-1}\rangle} , ~ i=1,2,3$, respectively
			and $\omega$ is a primitive cube root of unity in $\mathbb{F}_{p^m}^*$. Moreover,
			\[
			|\mathcal{C}| = |\mathcal{C}_1|\,|\mathcal{C}_2|\,|\mathcal{C}_3|, \quad 
			\mathcal{C}^\perp = \mathcal{C}_1^\perp \oplus \mathcal{C}_2^\perp \oplus \mathcal{C}_3^\perp,
			\]
			and  $\mathcal{C}$ is a self-dual code if and only if $\mathcal{C}_1 = \mathcal{C}_1^\perp$, $\mathcal{C}_2 = \mathcal{C}_3^\perp$, and $\mathcal{C}_3 = \mathcal{C}_2^\perp$.
		\end{enumerate}

\end{theorem}
\begin{proof}
Any skew $\theta$-cyclic code $\mathcal{C}$ of length $3p^s$ corresponds to a left ideal of $\mathcal{R}_1^{3p^s}$. Since $|\theta| \mid p^s$, the polynomial $x^{3p^s}-1$ remains central in $\mathbb{F}_{p^m}[x;\theta]$, and we have,
	\[
	x^{3p^s} - 1 = (x^{p^s} - 1)(x^{2p^s} + x^{p^s} + 1).
	\]
	If $m$ is odd, the polynomial $g(x)=x^2 +x + 1$ is irreducible over $\mathbb{F}_{p^m}$.  Because if this is reducible then the roots of $g(x)$ are $\frac{-1 \pm \sqrt{-3}}{2}$. Since $p \equiv 2 \mod 3$, $-3$ is a non-quadratic residue in $\mathbb{F}_p$, meaning $\sqrt{-3} \in \mathbb{F}_{p^2} \setminus \mathbb{F}_p$ \cite{Bur10}. Thus, $\sqrt{-3} \in \mathbb{F}_{p^m}$ if and only if $2 \mid m$. With $m$ odd, the roots lie outside $\mathbb{F}_{p^m}$.
	Consequently, the factorization is
	$
	x^{3p^s}-1=(x^{p^s}-1)(x^{2p^s}+x^{p^s}+1).
	$
	Since \(\mathbb{F}_p\) is fixed by \(\theta\) and \(|\theta|\mid p^s\),
	it follows from Proposition \ref{Central poly generalization} that both
	$
	x^{p^s}-1
	\quad \text{and} \quad
	x^{2p^s}+x^{p^s}+1
	$
	are central polynomials.
Hence,	the generated ideals are two-sided in $\mathbb{F}_{p^m}[x; \theta]$. By the CRT, the ambient ring decomposes as:
	\[ \mathcal{R}^{3p^s}_{1} \cong \frac{\mathbb{F}_{p^m}[x; \theta]}{\langle x^{p^s} - 1 \rangle} \oplus \frac{\mathbb{F}_{p^m}[x; \theta]}{\langle x^{2p^s} + x^{p^s} + 1 \rangle}. \]
	Hence, every skew $\theta$-cyclic code $\mathcal{C}$  of length $3p^s$ over $\mathbb{F}_{p^m}$ can be written as,  
	$\mathcal{C} = \mathcal{C}_1 \oplus \mathcal{C}_2,$
	where $\mathcal{C}_1$ and $\mathcal{C}_2$
	are left ideals of $\frac{\mathbb{F}_{p^m}[x;\theta]}{\langle  x^{p^s}-1\rangle}$ and $\frac{\mathbb{F}_{p^m}[x;\theta]}{\langle  x^{2p^s}+ x^{p^s}+1\rangle}$, respectively.
As 	$\mathcal{C}_1^\perp$ is clear,
for  $\mathcal{C}_2^\perp$, let
\[
\mathcal{R}^{3p^s}_{k, f}=\frac{R_k[x;\Theta]}{\langle f(x)\rangle}, 
\quad f(x)=x^{2p^s} +x^{p^s}+1
\]
As $\mathcal{C}_2$ be a left ideal of $\mathcal{R}^{3p^s}_{k,f}$. Its right annihilator is
$
\mathcal{A}(\mathcal{C}_2)=\{\,h(x)\in \mathcal{R}^{3p^s}_{k,f} \mid g(x)h(x)=0,\ \forall\, g(x)\in \mathcal{C}_2\,\}. 
$ Hence the annihilator consists of those polynomials whose products vanish modulo $f(x)$.
Now, $( x^{2p^s}+x^{p^s}+1)^*= x^{2p^s}+x^{p^s}+1.$ 
Therefore, by Remark \ref{remark 1}, $\mathcal{C}_2^\perp $ is also left ideal of $\mathcal{R}^{3p^s}_{k,f}.$

	Now if, $m$ is even, $\mathbb{F}_{p^2} \subseteq \mathbb{F}_{p^m}$. As shown above, $\sqrt{-3} \in \mathbb{F}_{p^2}$, so $-3$ is a quadratic residue in $\mathbb{F}_{p^m}$. The polynomial $y^2 + y + 1$ splits into linear factors, and we have:
	\[ x^{3p^s} - 1 = (x^{p^s} - 1)(x^{p^s} - \omega)(x^{p^s} - \omega^2), \]
	where $\omega$ is a primitive cube root of unity in $\mathbb{F}_{p^m}$.

\end{proof}


	Now we will study the algebraic structure  of skew constacyclic codes of length $3p^s$ over $R_k$.
\subsection{Skew constacyclic codes of length $3p^s$ over $R_k, k>1$}
 For  ring $R_k$ with $k > 1$ and  $\lambda \in R_k^*$, 
  \Cref{lemma_uniform_iso} is no longer applicable because $\lambda^{p^m} \neq \lambda$ in general.  Consequently, even if there exists a unit $\delta \in R_k^*$ such that $\delta^3 = \lambda$, the skew $(\lambda, \Theta)$-constacyclic codes of length $3p^s$ ($p$ odd prime coprime with $3$) need not be  equivalent to the skew $\Theta$-cyclic codes of  same length via the  isomorphism $\Phi$. \\
 Let \( \Theta \) be an automorphism of \( R_k \), as defined in \Cref{auto R_k}, satisfying \( \eta_j = 1 \) for all \( 1 \le j < k-1 \), that is, for every \( \sum a_i u^i \in R_k \), we have
 \[
 \Theta\left(\sum a_i u^i\right)=\sum \theta(a_i)u^i.
 \]
 Assume that the order of \( \Theta \), denoted by \( |\Theta| \), divides \( p^s \). Suppose further that \( \lambda \in (R_k^*)^{p^s} \). This implies there exist $\lambda_0 \in R_k^*$ such that $\lambda_0^{p^s}=\lambda$ 
 satisfies \( \Theta(\lambda_0)=\lambda_0 \).

	We analyze the structure of skew $(\lambda, \Theta)$-constacyclic codes by considering whether the unit $\lambda_0 \in R_k$ is a perfect cube in $R_k$.

	Let us first assume $\lambda_0$ is a perfect cube in $R_k^*$ i.e.,    $\lambda_0=\delta^3$, for some $\delta \in R_k^*$.
		Write $\lambda_0 =\lambda_1+\lambda_2u+\ldots+\lambda_{k}u^{k-1}$ such that $\lambda_1 \in \mathbb{F}_{p^m}^*, \lambda_i \in \mathbb{F}_{p^m}, 2 \leq i \leq k$ and $\delta=\delta_1+\delta_2u+\ldots+\delta_ku^{k-1}$, where $\delta_1 \in \mathbb{F}_{p^m}^*, \delta_i \in \mathbb{F}_{p^m}, 2 \leq i \leq k$. 
			\begin{lemma}\label{lemma fixing lambda_0}
			Let $R_k = \mathbb{F}_{p^m}[u]/\langle u^k \rangle$ where $p \neq 3$. Let $\Theta$ be an automorphism of $R_k$ such that the order of $\Theta$ divides $p^s$. If $\lambda_0 \in R_k$ is a unit such that $\Theta(\lambda_0) = \lambda_0$ and $\lambda_0 = \delta^3$ for some $\delta \in R_k^*$, then $\Theta(\delta) = \delta$.
		\end{lemma}
		\begin{proof}
			As given $\lambda_0=\delta^3$ for some $\delta \in R_k^*$. This implies $\lambda_1+\lambda_2u+\ldots+\lambda_{k}u^{k-1}=(\delta_1+\delta_2u+\ldots+\delta_ku^{k-1})^3,$ where $\lambda_1, ~\delta_1  \in \mathbb{F}_{p^m}^*$ and $ \lambda_i, \delta_i \in \mathbb{F}_{p^m}, 2 \leq i \leq k$. 
By simple computations, we obtain
			$\lambda_n = \sum_{\substack{1 \leq i,j,l \leq n \\ i + j + l = n + 2}} \delta_i \delta_j \delta_l, 1 \leq n \leq k$,
			as we have to multiply all combinations $\delta_iu^{i-1}, \delta_ju^{j-1}, \delta_lu^{l-1}$ such that
			$i-1+j-1+l-1=n-1$. This implies, \begin{equation}\label{eq lambda}
				\lambda_0=\sum_{n=1}^{k}(\sum_{\substack{1 \leq i,j,l \leq n \\ i + j + l = n + 2}} \delta_i \delta_j \delta_l)u^n.
			\end{equation}
			By compairing the cofficient of $1$ both side, we obtain $\lambda_1=\delta_1^3, \delta_1 \in \mathbb{F}_{p^m}^*$.
			This gives $\lambda_0$ is perfect cube in $R_k^*$ if and only if $\lambda_1$ is perfect cube in $\mathbb{F}_{p^m}^*$.  Now because $\Theta(\lambda_0)=\lambda_0$, this gives, $\theta(\delta_1^3)=\delta_1^3.$ From \Cref{Pro 3p^s iso}, we get $\theta(\delta_1)=\delta_1$. Now on compairing the cofficient of $u$ in (\ref{eq lambda}) both side, we obtain, $3\delta_1^2\delta_2=\lambda_2.$ By the given condition, $\theta(3\delta_1^2\delta_2)=3\delta_1^2\delta_2$, as $\theta(\delta_1)=\delta_1$ and $\delta_1$ is invertible, we get $\theta(\delta_2)=\delta_2$, repeating this  recursive process for higher powers of $u$
			 we obtain $\theta(\delta_i)=\delta_i$, for any non-zero $\delta_i$. Thus $\Theta(\delta)=\delta$.
		\end{proof}
		Now first assume  $p \equiv 1 \mod 3$.
\begin{theorem}\label{3p^s Rk thm}
	Suppose  $p \equiv 1 \mod{3}$ and $R_k = \mathbb{F}_{p^m}[u]/\langle u^k \rangle$ with $k>1$.
	Let $\Theta \in \mathrm{Aut}(R_k)$ be such that $|\Theta| \mid p^s$. Assume that $\lambda \in R_k^*$ is such that $\lambda=\lambda_0^{p^s} $, $\Theta(\lambda_0)=\lambda_0$ and 
$\lambda_0=\delta^3$ for some $\delta \in R_k^*$. Then every skew $(\lambda,\Theta)$-constacyclic code $\mathcal{C}$ of length $3p^s$ over $R_k$ decomposes as
	$
	\mathcal{C} = \mathcal{C}_1 \oplus \mathcal{C}_2 \oplus \mathcal{C}_3,
	$
	where $\mathcal{C}_i$ is a left ideal of
	$$
	\frac{R_k[x;\Theta]}{\langle (x- \omega^{i-1}\delta)^{p^s}  \rangle},  ~ i=1,2,3,
	$$
	$\omega \in \mathbb{F}_p^*$ being a primitive cube root.
		Moreover,
	$	|\mathcal{C}| = |\mathcal{C}_1|\,|\mathcal{C}_2|\,|\mathcal{C}_3|.
	$	Furthermore, the dual code satisfies
	$
	\mathcal{C}^\perp = \mathcal{C}_1^\perp \oplus \mathcal{C}_2^\perp \oplus \mathcal{C}_3^\perp,
	$
	where $\mathcal{C}_i^\perp$ is a left ideal of
	\[
	\frac{R_k[x;\Theta]}{\langle (x- \omega^{1-i}\delta^{-1})^{p^s}  \rangle}.
	\]
		In particular, $\mathcal{C}$ is self-dual if and only if
	$
	\mathcal{C}_1 = \mathcal{C}_1^\perp, \quad 
	\mathcal{C}_2 = \mathcal{C}_3^\perp, \quad 
	\mathcal{C}_3 = \mathcal{C}_2^\perp.
	$
\end{theorem}
\begin{proof}
	Using \Cref{lemma fixing lambda_0}, we have $\Theta(\delta)=\delta$.
	So, we can factorize $x^{3p^s} - \lambda$ over $R_k$.
As	\[ x^{3p^s} - \lambda = x^{3p^s} - \delta^{3p^s} = (x - \delta)^{p^s}(x^{2} + \delta x + \delta^2)^{p^s}. \]
Since	$p \equiv 1 \mod 3$, there exist element $\omega \in \mathbb{F}_p$ such that
$\omega^3 = 1.$ 
	The quadratic factor $x^{2} + \delta x + \delta^2$ can be further factorized using the roots $\omega$ and $\omega^2$. We can write,
$$ x^{2} + \delta x + \delta^2=
		(x - \omega \delta)(x- \omega^2 \delta).$$
	Thus,
	\[
	x^{3p^s}-\lambda
	=(x-\delta)^{p^s}(x-\omega\delta)^{p^s}(x-\omega^2\delta)^{p^s}.
	\]
	Since $\omega \in \mathbb{F}_p$, is fixed by the automorphism $\Theta$.  Therefore, the coefficients of these factors are fixed under $\Theta$, and the generated ideals are two-sided in the ring $R_k[x; \Theta]$.
		By the Chinese Remainder Theorem, the ambient ring decomposes as:
	\[ \frac{R_k[x; \Theta]}{\langle x^{3p^s} - \lambda \rangle} \cong \frac{R_k[x; \Theta]}{\langle (x - \delta)^{p^s} \rangle} \oplus \frac{R_k[x; \Theta]}{\langle (x - \omega \delta)^{p^s} \rangle} \oplus \frac{R_k[x; \Theta]}{\langle (x- \omega^2 \delta)^{p^s} \rangle}. \]
\end{proof}
	Consequently, any skew $(\lambda, \Theta)$-constacyclic code $\mathcal{C}$ of length $3p^s$ over $R_k$ can be uniquely written as a direct sum $\mathcal{C} = \mathcal{C}_1 \oplus \mathcal{C}_2 \oplus \mathcal{C}_3$, where $\mathcal{C}_i$ is skew $((\omega^i\delta)^{p^s}, \Theta)$-constacyclic codes of length $p^s$ with $i=1,2,3.$ Thus by using \Cref{ideals R_k}, we obtain the structure of corresponding ideals.\\
Now  we consider $p \equiv 2 \mod 3$.
\begin{theorem}
	Suppose $p \equiv 2 \mod{3}$ and let $\Theta \in \mathrm{Aut}(R_k)$ with $|\Theta| \mid p^s$. 
	Assume that $\lambda \in R_k^*$ is such that $\lambda=\lambda_0^{p^s} $, $\Theta(\lambda_0)=\lambda_0$ and 
	$\lambda_0=\delta^3$ for some $\delta \in R_k^*$.
	Let $\mathcal{C}$ be a skew $(\lambda, \Theta)$-constacyclic code of length $3p^s$ over $R_k$.
	
	\begin{enumerate}
		
		\item[(a)] If $m$ is odd, then $\mathcal{C}$ decomposes as
		$
		\mathcal{C} = \mathcal{C}_1 \oplus \mathcal{C}_2,
		$
		where $\mathcal{C}_1$ and $\mathcal{C}_2$ are left ideals of
		\[
		\frac{R_k[x;\Theta]}{\langle (x-\delta)^{p^s} \rangle}
		\quad \text{and} \quad
		\frac{R_k[x;\Theta]}{\langle (x^2+\delta x^{p^s}+\delta^2)^{p^s} \rangle},
		\]
		respectively. Moreover,
		$
		|\mathcal{C}| = |\mathcal{C}_1|\,|\mathcal{C}_2|, \quad 
		\mathcal{C}^\perp = \mathcal{C}_1^\perp \oplus \mathcal{C}_2^\perp,
		$
		where $\mathcal{C}_1^\perp$ and $\mathcal{C}_2^\perp$ are left ideals of
		\[
		\frac{R_k[x;\Theta]}{\langle (x-\delta^{-1})^{p^s} \rangle}
		\quad \text{and} \quad
		\frac{R_k[x;\Theta]}{\langle (x^2+\delta^{-1} x+\delta^{-2})^{p^s} \rangle},
		\]
		respectively. In particular, $\mathcal{C}$ is self-dual if and only if 
		$
		\mathcal{C}_1 = \mathcal{C}_1^\perp \quad \text{and} \quad \mathcal{C}_2 = \mathcal{C}_2^\perp.
		$
		
		\item[(b)] If $m$ is even, then $\mathcal{C}$ decomposes as
		$
		\mathcal{C} = \mathcal{C}_1 \oplus \mathcal{C}_2 \oplus \mathcal{C}_3,
		$
		where $\mathcal{C}_i$ is left ideal of
		\[
		\frac{R_k[x;\Theta]}{\langle (x-\omega^{\,i-1}\delta)^{p^s} \rangle}, ~ i=1,2,3,
		\]
		and $\omega \in \mathbb{F}_{p^m}^*$ is a primitive cube root of unity. Moreover,
		$
		|\mathcal{C}| = |\mathcal{C}_1|\,|\mathcal{C}_2|\,|\mathcal{C}_3|, \quad 
		\mathcal{C}^\perp = \mathcal{C}_1^\perp \oplus \mathcal{C}_2^\perp \oplus \mathcal{C}_3^\perp,
		$
		where $\mathcal{C}_i^\perp$ is left ideal of
		\[
		\frac{R_k[x;\Theta]}{\langle (x-\omega^{\,1-i}\delta^{-1})^{p^s} \rangle}, ~ i=1,2,3.
		\]
		In particular, $\mathcal{C}$ is self-dual if and only if 
		$
		\mathcal{C}_1 = \mathcal{C}_1^\perp, \quad 
		\mathcal{C}_2 = \mathcal{C}_3^\perp, \quad 
		\mathcal{C}_3 = \mathcal{C}_2^\perp.
		$
		
	\end{enumerate}
\end{theorem}
\begin{proof}
	Using \Cref{lemma fixing lambda_0}, we can factorize $x^{3p^s} - \lambda$ over $R_k$.
	As 	\[ x^{3p^s} - \lambda = x^{3p^s} - \delta^{3p^s} = (x - \delta)^{p^s}(x^{2} + \delta x + \delta^2)^{p^s}. \]
	First we consider $m$ is odd. We prove that in this case $x^{2} + \delta x + \delta^2$ is irreducible over $R_k$.
		Suppose, for the sake of contradiction, that $ x^2 + \delta x + \delta^2$ is reducible over $R_k$. Then there exists a root $r \in R_k$ such that $r^2 + \delta r + \delta^2 = 0. $
	This implies, $\mu_k(r^2 + \delta r + \delta^2) = 0 $ in $\mathbb{F}_{p^m}.$ Let $\bar{r}$ and $\bar{\delta}$ be the images of $r$ and $\delta$ respectively. We have $\bar{r}^2 + \bar{\delta} \bar{r} + \bar{\delta}^2 = 0.$ This implies,
	$ \bar{r}^3 - \bar{\delta}^3 = 0 \implies (\bar{r}\bar{\delta}^{-1})^3 = 1. $ If $\bar{r}\bar{\delta}^{-1}=1$, then $\bar{r} = \bar{\delta}$. Thus we have, $\bar{\delta}^2 + \bar{\delta}^2 + \bar{\delta}^2 = 3\bar{\delta}^2 = 0. $ Since $\lambda$ is a unit, $\bar{\delta}\neq 0$. Thus, $3 = 0$, which contradicts $p \neq 3$.
	If $\bar{r}\bar{\delta}^{-1} \neq 1$, then the order of $ \bar{r}\bar{\delta}^{-1}$ is exactly $3$. For an element of order $3$ to exist in $\mathbb{F}_{p^m}$, we must have $3 \mid |\mathbb{F}_{p^m}^*|$, i.e., $3 \mid p^m - 1$.
	This implies	$p^m \equiv 1 \mod 3. $
	However, we are given $p \equiv -1 \mod 3$. Since $m$ is odd, we have
	$p^m \equiv (-1)^m \equiv -1  \mod 3.$
	This contradicts $p^m \equiv 1 \mod 3$.
	This implies, $x^{2} + \delta x + \delta^2$ is irreducible over $R_k$.
		Since $\delta$ is fixed by the automorphism $\Theta$. Therefore, by \Cref{Central poly generalization} the generated ideals are two-sided in  $R_k[x; \Theta]$.
	By the Chinese Remainder Theorem, the ambient ring decomposes as:
	\[ \frac{R_k[x; \Theta]}{\langle x^{3p^s} - \lambda \rangle} \cong \frac{R_k[x; \Theta]}{\langle (x - \delta)^{p^s} \rangle} \oplus \frac{R_k[x; \Theta]}{\langle (x^{2} + \delta x + \delta^2)^{p^s} \rangle}. \]
	Consequently, any skew $(\lambda, \Theta)$-constacyclic code $\mathcal{C}$ of length $3p^s$ over $R_k$ can be uniquely written as a direct sum $\mathcal{C} = \mathcal{C}_1 \oplus \mathcal{C}_2$, where $\mathcal{C}_1$ and $ \mathcal{C}_2$  are skew $(\delta^{p^s}, \Theta)$-constacyclic  and skew $(f, \Theta)$-polycyclic
	codes of length $p^s$ and $2p^s$ respectively, where $f= x^{2p^s} + \delta^{p^s} x^{p^s} + \delta^{2p^s}$. \\
	Since $(x^{2} + \delta x + \delta^2)^*=\delta^2x^{2}+\delta x+1= \delta^2(x^{2}+\delta^{-1}x+\delta^{-2}).$ Because $\Theta(\delta^{-1})=\delta^{-1}$ and $|\Theta| \mid p^s$; 
	 the ideal $\langle x^{2p^s}+\delta^{-1}x^{p^s}+\delta^{-2} \rangle $ is central in $R_k[x; \Theta]$. Hence dual is clear from \Cref{remark 1}.\\
	Now let us consider  $m$ is even.
	Since $p \equiv -1 \mod 3$, we have,
	$p^m \equiv (-1)^m \mod 3,$
 thus $p^m \equiv 1 \mod 3$. Consequently, $\mathbb{F}_{p^m}$ contains primitive cube roots of unity, denoted by $\omega$ and $\omega^2$. 
 Now because of order of $\Theta$ divides $p^s$  and $p$ is odd prime strictly bigger then $3$, so $\Theta(\omega)=\omega$.
 Thus,  $x^{2p^s} + \delta x^{p^s} + \delta^2=
	(x^{p^s} - \omega\delta)(x^{p^s} - \omega^2\delta). $
 Hence,
	$x^{3p^s} - \lambda = (x^{p^s} - \delta)(x^{p^s} - \omega\delta)(x^{p^s} - \omega^2\delta). $
	The ambient ring decomposes as
	$$\frac{R_k[x; \Theta]}{\langle x^{3p^s} - \lambda \rangle} \cong \frac{R_k[x; \Theta]}{\langle x^{p^s} - \delta\rangle} \oplus \frac{R_k[x; \Theta]}{\langle x^{p^s} - \omega\delta \rangle} \oplus \frac{R_k[x; \Theta]}{\langle x^{p^s} - \omega^2\delta \rangle}. $$
	\end{proof}
	
	Assume that there exists $\lambda_0 \in R_k^*$ such that $\lambda=\lambda_0^{p^s}$ and $\Theta(\lambda_0)=\lambda_0$.
	Then $x^{3p^s}-\lambda=(x^3-\lambda_0)^{p^s}.$
		If $\lambda_0$ is not a cube in $R_k^*$, then by the proof of  \Cref{lemma fixing lambda_0}, $\mu_k(\lambda_0)$ is also not a cube in $\mathbb{F}_{p^m}$. Hence,
	$x^3-\mu_k(\lambda_0)$
	is irreducible in $\mathbb{F}_{p^m}[x;\theta]$, which implies that $x^3-\lambda_0$ is irreducible in $R_k[x;\Theta]$. Therefore, the structure of the left ideals follows from \Cref{ideals R_k}.

  \section{Skew cyclic and skew negacyclic codes of length $6p^s$ over $R_k, k \geq 1$}
 In this section, we study the algebraic structure of  skew cyclic and skew negacyclic code of length $6p^s$ over $R_k$.  Consider $\mathcal{R}_{k, \lambda}^{6p^s} = \frac{R_k[x;\Theta]}{\langle x^{6p^s}-\lambda \rangle}$, where $\lambda \in \{1, -1\}$ represent skew cyclic and skew negacyclic code, respectively.
 
 \subsection{Structure of Skew Cyclic Codes of Length $6p^s$ over $R_k$}\label{sub 5.1}
 Note that $x^{6p^s}-1=(x^6-1)^{p^s}$ and  $x^6-1=(x-1)(x+1)(x^2-x+1)(x^2+x+1)$ over $R_k$. 
Further, if $p \equiv 1 \mod 6$ or $m$ is even then  $\mathbb{F}_{p^m}$ contains a primitive $6^{th}$ roots of unity, say $\alpha$. Thus $\Theta(\alpha)=\alpha$ because if $\Theta(\alpha)=\alpha^{-1}$ then it contradict the assumption, $|\Theta| \mid p^s$.
 Thus,
 $$x^{6p^s}-1 =(x-1)^{p^s}(x+1)^{p^s}(x-\alpha^2)^{p^s}(x-\alpha^4)^{p^s}(x-\alpha)^{p^s}(x-\alpha^5)^{p^s}.$$
 Hence,
 \[ \mathcal{R}_k^{6p^s} \cong \bigoplus_{i=1}^6 \frac{R_k[x;\Theta]}{\langle (x^{p^s}-\alpha^{i{p^s}}) \rangle} .\]
 Every skew $\Theta$-cyclic code $\mathcal{C}$ of length $6p^s$ is a direct sum $\mathcal{C} = \bigoplus_{i=1}^6 \mathcal{C}_i$, where each $\mathcal{C}_i$ is skew $(\Theta, \alpha^{i{p^s}})$-constacyclic code of length $p^s$.\\
Now, if $p \equiv 5 \mod 6$ and $m$ is odd
 then $x^2 \pm x + 1$ are irreducible over $R_k$. Indeed, if $x^2 \pm x +1$ is reducible over $R_k$, then there exists $\gamma \in R_k$, where $\gamma=\sum_{i=0}^{k-1}u^i\gamma_i$ with $\gamma_i \in \mathbb{F}_{p^m}$, such that $\gamma^2 \pm \gamma +1=0$. Applying $\mu_k$, we get $\mu_k(\gamma)^2 \pm \mu_k(\gamma)+1=0$. Hence, $\sqrt{-3}\in \mathbb{F}_{p^m}$, a contradiction. Therefore, $x^2 \pm x +1$ are irreducible over $R_k$.
   Thus, Proposition \ref{Central poly generalization}, we have
 \[ \mathcal{R}_k^{6p^s} \cong \frac{R_k[x;\Theta]}{\langle (x-1)^{p^s} \rangle} \oplus \frac{R_k[x;\Theta]}{\langle (x+1)^{p^s} \rangle} \oplus \frac{R_k[x;\Theta]}{\langle (x^2-x+1)^{p^s} \rangle} \oplus \frac{R_k[x;\Theta]}{\langle (x^2+x+1)^{p^s} \rangle}. \]
Hence, every skew \(\Theta\)-cyclic code \(\mathcal{C}\) can be written as
$\mathcal{C}=\mathcal{C}_1\oplus \mathcal{C}_2\oplus \mathcal{C}_3\oplus \mathcal{C}_4,
$
where \(\mathcal{C}_1\) is a skew \(\Theta\)-cyclic code of length \(p^s\), \(\mathcal{C}_2\) is a skew \(\Theta\)-negacyclic code of length \(p^s\) while  \(\mathcal{C}_3\) and \(\mathcal{C}_4\) are skew \((f_1,\Theta)\)-polycyclic and skew \((f_2, \Theta)\)-polycyclic codes of length \(2p^s\), respectively, where $f_1=x^{2p^s}-x^{p^s}+1 $ and $f_2=x^{2p^s}+x^{p^s}+1.$
 \\
 
 Now we move further to study skew negacyclic codes of length $6p^s.$
\subsection{Structure of Skew Negacyclic Codes of Length $6p^s$ over $R_k$} 
This subsection is divided into the following four cases based on the congruence class of the prime \( p \) modulo \( 12 \).
By Proposition~\ref{lemma_uniform_iso}, the cases \(p\equiv 1, 5, 7, \text{and}~ 11 \mod{12}\) with \(m\) even are equivalent to the corresponding skew cyclic cases, since in each of these cases there exist $\alpha' \in \mathbb{F}_{p^m}^*$ such that $\alpha'^6=-1$.
	Hence, these cases are omitted.\\
Now, we consider the remaining cases.
\subsubsection{\boldmath{$p \equiv 5 \mod 12$ and $m$ odd}}
Since $12~|~p-5$, we have $4~|~p-1$ but $12~\nmid p-1$ then there exist an element $\alpha$ in $\mathbb{F}_p$ such that
\[
\alpha= \xi^{\frac{p-1}{4}} ~\text{and }~\alpha^{-1}= \xi^{\frac{3(p-1)}{4}},
\] where $\xi$ be a generator of $\mathbb{F}_p^{*}=\mathbb{F}_p \setminus \{0\}.$ For any $\Theta \in $ Aut$(R_k)$, we have $\Theta(\alpha)=\alpha,$ $\Theta(\alpha^{-1})=\alpha^{-1}.$\\
Therefore, the polynomial $x^{6}+1$ admits the following decomposition:
\begin{align*}
	(x^{6}+1)&=(x^{2}+1)(x^{4}-x^{2}+1)\\&
	=(x^{2}-\alpha^2) (x^{4}+(\alpha+\alpha^{-1})x^{3}+(\alpha \alpha^{-1}-2)x^{2}+1)\\&
	= (x-\alpha) (x+\alpha) (x^2+\alpha x -1)  (x^2+\alpha^{-1}x-1)
	\end{align*}

\begin{lemma}\label{ p 5 mod 12}
	If $m$ is odd then  $x^2+\alpha x -1, $ and  $x^2+\alpha^{-1} x-1$  are irreducible over $R_k.$
\end{lemma} 
\begin{proof} Suppose $x^2+\alpha x -1 $ is reducible over $R_k$, then there exist a root  $\gamma \in R_k$, where $\gamma=\sum_{i=0}^{k-1}u^i\gamma_i$, with $\gamma_i \in \mathbb{F}_{p^m}$, such that $	\gamma^2+\alpha\gamma-1=0$.
	This implies $ \mu_k(\gamma)^2+\alpha \mu_k(\gamma)^2-1=0$, so
	$\mu_k(\gamma)=\frac{-\alpha \pm \sqrt{\alpha^2+4}}{2}.$
	Since $\alpha^2=-1$, we get $
	\mu_k(\gamma)=\frac{-\alpha \pm \sqrt{3}}{2}.
	$ As $3$ is  non-quadratic residue modulo $p,$ it follows that $\sqrt{3} \in \mathbb{F}_{p^2} \setminus \mathbb{F}_p.$ Hence, $\sqrt{3} \in \mathbb{F}_{p^m}$ if and only if $2~|~ m.$ But $m$ is odd, which gives a contradiction. Therefore,  $x^2+\alpha x -1 $ is an irreducible over $R_k.$ Similarly, we can prove $x^2+\alpha^{-1} x -1$ is irreducible over $R_k.$
\end{proof}
Hence, by Proposition \ref{Central poly generalization}, and CRT, we have 
  \begin{align*}
 	\mathcal{R}_k^{6p^s}&\cong \frac{R_k[x; \Theta]}{\langle (x-\alpha)^{p^s} \rangle} \oplus \frac{R_k[x; \Theta]}{\langle (x+\alpha)^{p^s} \rangle} \oplus\frac{R_k[x; \Theta]}{\langle (x^{2}+\alpha x -1)^{p^s} \rangle} \oplus\frac{R_k[x; \Theta]}{\langle (x^{2}+\alpha^{-1} x-1)^{p^s} \rangle}.
 \end{align*}
 It follows that the left ideals of $ \mathcal{R}_k^{6p^s}$, say $\mathcal{C}$, is represented as a direct sum
 \[ \mathcal{C}=\mathcal{C}_1\oplus \mathcal{C}_2\oplus\mathcal{C}_3\oplus\mathcal{C}_4, \] where \(\mathcal{C}_1\) and \(\mathcal{C}_2\) are skew \((\alpha,\Theta)\)-constacyclic and skew \((-\alpha,\Theta)\)-constacyclic codes of length \(p^s\) over \(R_k\), respectively, while \(\mathcal{C}_3\) and \(\mathcal{C}_4\) are skew \((f_1, \Theta)\)-polycyclic and skew \((f_2, \Theta)\)-polycyclic codes of length \(2p^s\), respectively, where $f_1=x^{2p^s}+\alpha^{p^s} x^{p^s}-1
 \quad \text{and} \quad
 f_2=x^{2p^s}+\alpha^{-p^s}x^{p^s}-1.$\\
 
 Thus, we have following theorem:
 
  \begin{theorem}\label{6}
 	Let $\mathcal{C}=\mathcal{C}_1\oplus \mathcal{C}_2\oplus\mathcal{C}_3\oplus\mathcal{C}_4$ be a skew negacyclic code of length $6p^s$ over $R_k,$ where  $\mathcal{C}_1$, $\mathcal{C}_2$, $\mathcal{C}_3$, and $\mathcal{C}_4$  are the left ideals of $\frac{R_k[x; \Theta]}{\langle (x-\alpha)^{p^s} \rangle},$ $\frac{R_k[x; \Theta]}{\langle (x+\alpha)^{p^s} \rangle},$ $\frac{R_k[x; \Theta]}{\langle (x^{2}+\alpha x -1)^{p^s} \rangle},$ and $\frac{R_k[x; \Theta]}{\langle (x^{2}+\alpha^{-1} x-1)^{p^s} \rangle},$ respectively. Then $|\mathcal{C}|=\prod_{i=1}^{4} |\mathcal{C}_i|$ and its dual is given by
 	$$	\mathcal{C}^\perp=\mathcal{C}_1^\perp \oplus \mathcal{C}_2^\perp \oplus\mathcal{C}_3^\perp \oplus  \mathcal{C}_4^\perp, $$
 	 where $\mathcal{C}_1^\perp$ and  $\mathcal{C}_2^\perp$  are the left ideals of the ambient rings $\frac{R_k[x; \Theta]}{\langle (x+\alpha)^{p^s} \rangle}$ and
 	  $\frac{R_k[x; \Theta]}{\langle (x-\alpha)^{p^s} \rangle}$, while $\mathcal{C}_3^\perp$ and $\mathcal{C}_4^\perp$ are  left ideals of $\frac{R_k[x; \Theta]}{\langle (x^{2}+\alpha^{-1} x -1)^{p^s} \rangle} ,$ and  $\frac{R_k[x; \Theta]}{\langle (x^{2}+\alpha x -1)^{p^s} \rangle} ,$ respectively. Moreover, $\mathcal{C}$ is a self-dual code if and only if $\mathcal{C}_1=\mathcal{C}_2^\perp,$ $\mathcal{C}_2=\mathcal{C}_1^\perp,$ $\mathcal{C}_3=\mathcal{C}_4^\perp,$ and $\mathcal{C}_4=\mathcal{C}_3^\perp.$    
 	
 \end{theorem}

%
  \subsubsection{\boldmath{$p \equiv 7 \mod 12$ and $m$ odd}}
 Since $12~|~p-7$, we have  $6~|~p-1$, but $12~\nmid p-1.$ Then there exist an element $\alpha$ in $\mathbb{F}_p$ such that
 $
 \alpha= \xi^{\frac{p-1}{6}} \text{and }~~\alpha^{-1}= \xi^{\frac{5(p-1)}{6}},
 $ where $\xi$ is a generator of $\mathbb{F}_p^{*}=\mathbb{F}_p \setminus \{0\}.$  Moreover, for all $\Theta \in Aut(R_k)$, we have $\Theta(\alpha)=\alpha$, $\Theta(\alpha^{-1})=\alpha^{-1}.$\\
 Therefore, the polynomial $x^{6}+1$ admits the following decomposition:
 \begin{align*}
 	(x^6+1)&=(x^2+1)(x^4-x^2+1)\\&
 	=(x^2+1) (x^2-\alpha) (x^2-\alpha^{-1})
 \end{align*}
  Since $m$ is odd, then we have $p^m \equiv 7$ mod $12$.
 Suppose $x^2+1 $ is reducible over $R_k$, then there exist an element $\gamma$ in $R_k$ such that $\gamma^2=-1.$  This further implies there exists $\mu_k(\gamma) \in \mathbb{F}_{p^m}$ such that $\mu_k(\gamma)^2=-1$. Since $p \equiv 7 $ mod $12,$ we have $p \equiv 3 $ mod $4$, then $-1$ is non-quadratic residue modulo $p$ and $\sqrt{-1}\in \mathbb{F}_{p^2} \setminus \mathbb{F}_p$. Thus $\sqrt{-1}\in \mathbb{F}_{p^m}$ if and only if $2~|~m.$ But $m$ is odd so,  $\sqrt{-1} \not \in \mathbb{F}_{p^m}$, which gives a contradiction.
Similarly if $x^2-\alpha=0$ is reducible over $R_k$ then there exist an element $\gamma \in R_k$ such that $\gamma^2-\alpha=0$, which  implies there exists some $\mu_k(\gamma) \in \mathbb{F}_{p^m}$ such that $\mu_k(\gamma)^2-\alpha=0$.   Hence $\mu_k(\gamma)^{12}=1$  $i.e.$ $12~|~ (p^m-1)$, which leads to a contradiction as $p^m \equiv 7$ mod $12.$ Therefore,  $x^2-\alpha$ is irreducible over $R_k.$ Similarly,  $x^2-\beta$ is irreducible over $R_k.$ Therefore,
 \begin{align*}
 	(x^{6}+1)^{p^s}=(x^{2}+1)^{p^s} (x^{2}-\alpha)^{p^s} (x^{2}-\alpha^{-1})^{p^s}
 \end{align*}
 Hence, by Proposition \ref{Central poly generalization}, we have 
 \begin{align*}
 	\mathcal{R}_k^{6p^s}\cong \frac{R_k[x; \Theta]}{\langle (x^{2}+1)^{p^s}\rangle} \oplus \frac{R_k[x; \Theta]}{\langle (x^{2}-\alpha)^{p^s}\rangle} \oplus\frac{R_k[x; \Theta]}{\langle (x^{2}-\alpha^{-1})^{p^s}\rangle}.
 \end{align*}
It follows that the left ideals of $ \mathcal{R}_k^{6p^s}$, say $\mathcal{C}$, is represented as a direct sum $\mathcal{C}=\mathcal{C}_1\oplus \mathcal{C}_2\oplus\mathcal{C}_3$, where \(\mathcal{C}_1\), \(\mathcal{C}_2\) and \(\mathcal{C}_3\) are skew \(\Theta\)-negacyclic, skew \((\alpha^{p^s}, \Theta)\)-constacyclic and skew \((\alpha^{-p^s}, \Theta)\)-constacyclic codes of length \(2p^s\) over \(R_k\), respectively and its dual is given by $\mathcal{C}^\perp=\mathcal{C}_1^\perp \oplus \mathcal{C}_2^\perp \oplus\mathcal{C}_3^\perp,$   where $\mathcal{C}_1^\perp,$  $\mathcal{C}_2^\perp,$ and $\mathcal{C}_3^\perp$ are the left ideals of the rings $\frac{R_k[x; \Theta]}{\langle (x^{2}+1)^{p^s} \rangle}, \frac{R_k[x; \Theta]}{\langle(x^{2}-\alpha^{-1})^{p^s} \rangle}$, and $\frac{R_k[x; \Theta]}{\langle(x^{2}-\alpha)^{p^s}\rangle}.$  Moreover, $\mathcal{C}$ is a self-dual code if and only if $\mathcal{C}_1=\mathcal{C}_1^\perp,$ $\mathcal{C}_2=\mathcal{C}_3^\perp,$ and $\mathcal{C}_3=\mathcal{C}_2^\perp.$

   \subsubsection{\boldmath{$p \equiv 11 \mod 12$ and $m$ odd}} Since $3$ is a quadratic residue modulo $p$, there exist an element $\beta$ in $\mathbb{F}_{p}$ such that $\beta^2=3.$ Therefore, the polynomial $x^{6}+1$ admits the following decomposition:
 \[
 x^6+1=(x^2+1)(x^2+\beta x+1)(x^2-\beta x+1).
 \]
  Since $m$ is odd and $p\equiv 3\mod{4}$, we have $-1$ is not a quadratic residue modulo $p.$ Therefore, by similar explanation as given in lemma \ref{ p 5 mod 12} we have $x^2\pm \beta x+1$ and  $x^2+1$ are irreducible polynomial over $R_k.$
 Hence by Proposition \ref{Central poly generalization} and CRT, we have
 \[	\mathcal{R}_k^{6p^s} \cong \frac{R_k[x; \Theta]}{\langle (x^{2}+1)^{p^s} \rangle}  \oplus\frac{R_k[x; \Theta]}{\langle(x^{2}+\beta x +1)^{p^s} \rangle} \oplus\frac{R_k[x; \Theta]}{\langle(x^2-\beta x+1)^{p^s} \rangle}.
 \]
 It follows that the left ideals of $ \mathcal{R}_k^{6p^s}$, say $\mathcal{C}$, is represented as a direct sum
 $\mathcal{C}=\mathcal{C}_1\oplus \mathcal{C}_2\oplus\mathcal{C}_3$ where \(\mathcal{C}_1\) is a skew \(\Theta\)-negacyclic code of length \(2p^s\) over \(R_k\), while \(\mathcal{C}_2\) and \(\mathcal{C}_3\) are skew \((f_1, \Theta)\)-polycyclic and skew \((f_2, \Theta)\)-polycyclic codes of length \(2p^s\) over \(R_k\), respectively, where
 \[
 f_1=x^{2p^s}+\beta^{p^s} x^{p^s}+1
 ~ \text{and} ~
 f_2=x^{2p^s}-\beta^{p^s} x^{p^s}+1.
 \] Its dual is given by
 $	\mathcal{C}^\perp=\mathcal{C}_1^\perp \oplus \mathcal{C}_2^\perp \oplus\mathcal{C}_3^\perp, $
 where $\mathcal{C}_1^\perp$, $\mathcal{C}_2^\perp,$ and
 $\mathcal{C}_3^\perp$ are the left ideals of the ambient ring $\frac{R_k[x; \Theta]}{\langle (x^2+1)^{p^s}\rangle},$    $\frac{R_k[x; \Theta]}{\langle(x^2-\beta x +1)^{p^s} \rangle},$ and  $\frac{R_k[x; \Theta]}{\langle(x^{2}+\beta x +1)^{p^s} \rangle},$ respectively. Moreover, $\mathcal{C}$ is a self-dual code if and only if $\mathcal{C}_1=\mathcal{C}_1^\perp,$ $\mathcal{C}_2=\mathcal{C}_3^\perp,$ and $\mathcal{C}_3=\mathcal{C}_2^\perp$.
 
In the next section, we present examples of MDS codes constructed using the theory developed above.
 \section{Examples of MDS Skew Constacyclic Codes Over $\mathbb{F}_{p^m}$} 
 A linear code $\mathcal{C}$ over $\mathbb{F}_{p^m}$ of length $n$, dimension $k^{'}$, and minimum distance $d$ is referred to as an $[n, k^{'}, d]$-code. 
A linear $[n,k',d]$-code is called an MDS code if it satisfies $d=n-k'+1$, and such codes are also known as optimal codes. In this section, we illustrate the above developed theory through several examples of optimal (MDS) skew constacyclic codes.\\
We first construct some examples of MDS codes corresponding to \Cref{3p^s first thm}, i.e., in the case when $p \equiv 1 \mod{3}$.
\subsection{\boldmath{$p \equiv 1 \mod{3}$}}
We consider the finite field \( \mathbb{F}_{7^{7}} \) and let
\(\theta \in \mathrm{Aut}(\mathbb{F}_{7^{7}})\) be the automorphism defined by $\theta(a)=a^{7}, ~ \forall a \in \mathbb{F}_{7^{7}}.$

The primitive cube roots of unity in \( \mathbb{F}_{7}^{*} \) are
$
\xi=2, ~ \xi^{-1}=4.
$
Let us consider the code $\mathcal{C}$ of length $21$. Observe that
$
|\theta|=7,
$
which satisfies the condition
$
|\theta| \mid 7.
$According to \Cref{3p^s first thm} every skew \( \theta \)-cyclic code
\( \mathcal{C} \) of length \(21\) over \( \mathbb{F}_{7^{7}} \)
admits the decomposition
\[
\mathcal{C}=\mathcal{C}_{1}\oplus \mathcal{C}_{2}\oplus \mathcal{C}_{3},
\]
where
\[
\mathcal{C}_{i}
\text{ is a left ideal of }
\frac{\mathbb{F}_{7^{7}}[x;\theta]}
{\langle x^{7}-2^{i-1}\rangle}, i=1, 2,3.
\]
According to \Cref{ideal over F_{p^m}}, $\mathcal{C}_i$ is generated by $f_i(x)$, where $f_i(x)$ is factor of $x^{7}-2^{i-1}.$
 Let $\omega$ be a $(7^7-1)^{th}$ primitive root of an irreducible polynomial of degree $7$ in $\mathbb{F}_7[x]$. Using SageMath, we factor the polynomial \( x^{7}-2 \) in the skew
polynomial ring \( \mathbb{F}_{7^{7}}[x;\theta] \); is given by,
	$$x^{7}-2=(x+a_1)(x+a_2)(x+a_3)(x+a_4)(x+a_5)(x+a_6)(x+a_7)$$ where,

\[
\begin{aligned}
	a_1 &= 5w^6 + w^5 + 2w^4 + w^3 + 3w^2 + 6w + 2, \\
	a_2 &= 4w^5 + 4w^3 + w^2 + 1, \\
	a_3 &= 4w^6 + 6w^4 + w^3 + 6w + 2, \\
	a_4 &= 5w^6 + 5w^5 + 3w^4 + 5w^3 + 2w^2 + 3w + 2, \\
	a_5 &= 4w^6 + 6w^5 + 6w^4 + 3w^3 + 2w^2 + 4w + 3,\\
	a_6 &=6w^5 + 6w^4 + 2w^3 + 3w^2 + 4w + 6,\\
	a_7 &=3w^6 + 2w^5 + 4w^3 + 4w + 2.
\end{aligned}
\]
By using the MAGMA\cite{BCP97} algebra system, we obtain the MDS skew $(2,\theta)$-constacyclic codes of length $7$ generated by the factors of $x^{7}-2$ over $\mathbb{F}_{7^7}$, as listed in \Cref{tab:x7minus2codes}.
\begin{table}[ht!]
	\centering
	\renewcommand{\arraystretch}{1.3}
	
	\begin{tabular}{|c|c|c|c|}
		\hline
		Generator polynomial 
		& $[n,k,d]$ 
		& MDS 
		\\
		\hline
		
		$(x+a_1)(x+a_2)$
		& $[7,5,3]$
		& Yes
		\\
		\hline
		
		$(x+a_1)(x+a_2)(x+a_3)$
		& $[7,4,4]$
		& Yes
		\\
		\hline
		
		$(x+a_1)(x+a_2)(x+a_3)(x+a_4)$
		& $[7,3,5]$
		& Yes
		\\
		\hline
		
		$(x+a_1)(x+a_2)(x+a_3)(x+a_4)(x+a_5)$
		& $[7,2,6]$
		& Yes
		\\
		\hline
		
	\end{tabular}

		\caption{Parameters of skew constacyclic codes generated by factors of $x^{7}-2$.}
			\label{tab:x7minus2codes}
\end{table}
The factors of  the polynomial \(x^{7}-4\) in the skew
 polynomial ring \(\mathbb{F}_{7^{7}}[x;\theta]\) as given by,
  \[
 x^{7}-4
 =
 (x+a_1)(x+a_2)(x+a_3)(x+a_4)(x+a_5)(x+a_6)(x+a_7),
 \]
 where
 \[
 \begin{aligned}
 	a_1 &= 5w^6 + w^4 + 3w + 2, \\
 	a_2 &= w^6 + 3w^5 + w^4 + 3w^3 + 6w^2 + 2w + 5, \\
 	a_3 &= 3w^6 + w^5 + 4w^4 + 5w^3 + 4w^2 + 3, \\
 	a_4 &= 2w^6 + 4w^5 + 5w^4 + 5w^3 + 4w^2 + 3w + 2, \\
 	a_5 &= 6w^5 + w^4 + w^3 + w^2 + 3w + 4, \\
 	a_6 &= 4w^6 + 3w^5 + 3w^4 + 6w^3 + 2w + 1, \\
 	a_7 &= 6w^6 + w^5 + 6w^4 + 4w^3 + 6.
 \end{aligned}
 \]
The MDS skew $(4,\theta)$-constacyclic codes of length $7$ generated by the factors of $x^{7}-4$ over $\mathbb{F}_{7^7}$, as listed in \Cref{tab:x7minus4}.
\begin{table}[ht!]
	\centering

	\renewcommand{\arraystretch}{1.2}
	\setlength{\tabcolsep}{4pt}
	\small
	\begin{tabular}{|c|c|c|}
		\hline
		Generator Polynomial  & $[n, k, d] $& MDS \\ 
		\hline
		
		$(x+a_{1})(x+a_{2})$
		& $7,5,3]$
		& Yes \\
		\hline
		
		$(x+a_{1})(x+a_{2})(x+a_{3})$
		& $[7,4,4]$
		& Yes \\
		\hline
		
		$(x+a_{1})(x+a_{2})(x+a_{3})(x+a_{4})$
		& $[7,3,5]$
		& Yes \\
		\hline
		
		$(x+a_{1})(x+a_{2})(x+a_{3})(x+a_{4})(x+a_{5})$
		& $[7,2,6]$
		& Yes \\
		\hline
			\end{tabular}
				\caption{Parameters of skew constacyclic codes generated by factors of $x^{7}-4$.}
	\label{tab:x7minus4}
\end{table}
The factors of  the polynomial \(x^{7}-1\) in the skew
polynomial ring \(\mathbb{F}_{7^{7}}[x;\theta]\) as given below,
\[
x^{7}-1
=
(x+a_1)(x+a_2)(x+a_3)(x+a_4)(x+a_5)(x+a_6)(x+a_7),
\]
where
\[
\begin{aligned}
	a_1 &= w^6 + 3w^5 + 5w^4 + 6w^3 + 5w^2 + 5w + 6, \\
	a_2 &= w^6 + w^4 + 4w^2 + 6w, \\
	a_3 &= 3w^6 + w^5 + 3w^4 + 6w^3 + w^2 + w, \\
	a_4 &= 3w^5 + w^3 + 3w^2 + 2w + 4, \\
	a_5 &= 2w^6 + 5w^5 + 6w^4 + 5w^2 + w + 6, \\
	a_6 &= 5w^6 + 2w^5 + 3w^4 + w^3 + 6w^2 + 3w + 3, \\
	a_7 &= 2w^6 + 2w^5 + 5w^4 + w^3 + 2w^2 + w + 2.
\end{aligned}
\]
The MDS skew $\theta$-cyclic codes of length $7$ generated by the factors of $x^{7}-1$ over $\mathbb{F}_{7^7}$, as listed in \Cref{tab:x7minus1}.
\begin{table}[h]
	\centering

	\renewcommand{\arraystretch}{1.2}
	\setlength{\tabcolsep}{4pt}
	\small
	\begin{tabular}{|c|c|c|}
		\hline
		Generator Polynomial & $[n,k,d]$ & MDS \\
		\hline
		
		$(x+a_{1})(x+a_{2})$
		& $[7,5,3]$
		& Yes \\
		\hline
		
		$(x+a_{1})(x+a_{2})(x+a_{3})$
		& $[7,4,4]$
		& Yes \\
		\hline
		
		$(x+a_{1})(x+a_{2})(x+a_{3})(x+a_{4})$
		& $[7,3,5]$
		& Yes \\
		\hline
		
		$(x+a_{1})(x+a_{2})(x+a_{3})(x+a_{4})(x+a_{5})$
		& $[7,2,6]$
		& Yes \\
		\hline
		
	\end{tabular}
		\caption{Parameters of skew cyclic codes generated by factors of $x^{7}-1$.}
	\label{tab:x7minus1}
\end{table}
Now we construct some examples of MDS codes corresponding to \Cref{3p^s 2nd theorem F} i.e., in the case when $p \equiv 2 \mod{3}$.
\subsection{\boldmath{$p \equiv 2 \mod{3}$}}
We consider the finite field \( \mathbb{F}_{5^{5}} \) and let
$\theta \in \mathrm{Aut}(\mathbb{F}_{5^{5}})$
be the  automorphism defined by
$
\theta(a)=a^{5}, ~ \forall a \in \mathbb{F}_{5^{5}}.
$

Let us consider the code of  length 
$15.$ Observe that
$|\theta|=5,$
which satisfies the condition
$
|\theta| \mid 5.
$
According to \Cref{3p^s 2nd theorem F} every skew \( \theta \)-cyclic code
\( \mathcal{C} \) of length \(15\) over \( \mathbb{F}_{5^{5}} \)
admits the decomposition
\[
\mathcal{C}=\mathcal{C}_{1}\oplus \mathcal{C}_{2},
\]
where
\[
\mathcal{C}_{1}
\text{ and }
\mathcal{C}_{2}
\text{ are left ideals of }
\frac{\mathbb{F}_{5^{5}}[x;\theta]}
{\langle x^{5}-1\rangle}
\text{ and }
\frac{\mathbb{F}_{5^{5}}[x;\theta]}
{\langle x^{10}+x^5+1\rangle},
\]
respectively.  Let $\omega$ be a $(5^5-1)^{th}$ primitive root of an irreducible polynomial of degree $5$ in $\mathbb{F}_5[x]$. 
The factorization of \(x^5-1\) in $\mathbb{F}_{5^5}[x; \theta]$ is given by
\[
\begin{aligned}
	x^5-1
	&=(x+a_1)(x+a_2)(x+a_3)(x+a_4)(x+a_5),
\end{aligned}
\]
where
\[
\begin{aligned}
	a_1&=4w^4+2w^3+2w^2+4w,\\
	a_2&=2w^4+2w^3+3w^2+w+4,\\
	a_3&=4w^4+4w^3+3w+2,\\
	a_4&=2w^4+w^3+4w^2+3w+2,\\
	a_5&=3w^4+2w^3+4w^2+2w+4.
\end{aligned}
\]
The MDS skew $\theta$-cyclic codes of length $5$ generated by the  factors of $x^5-1$ over \(\mathbb{F}_{5^5}\), as listed in \Cref{tab:skewcodes55}.
\begin{table}[ht]
	\centering
	\renewcommand{\arraystretch}{1.3}
	\begin{tabular}{|c|c|c|}
		\hline
		Generator polynomial & $[n,k,d]$ & MDS \\ 
		\hline
		
		$(x+a_1)(x+a_2)$
		& $[5,3,3]$& Yes \\
		\hline
		
		$(x+a_1)(x+a_2)(x+a_3)$
		& $[5,2,4]$ & Yes \\
	
		\hline
		
		$(x+a_3)(x+a_4)(x+a_5)$
		& $[5,2,4]$& Yes \\
		\hline
		
	\end{tabular}
	\caption{Parameters of skew cyclic codes generated by factors of $x^5-1$.}
			\label{tab:skewcodes55}
\end{table}
The factors of  the polynomial \(f(x)=x^{10}+x^5+1\) in the skew
polynomial ring \(\mathbb{F}_{5^5}[x;\theta]\) is given by,
\[
x^{10}+x^5+1
=
(x^2+ f_1)(x^2+f_2)(x^2+f_3)(x^2+f_4)(x^2+f_5),
\]
where
\[
\begin{aligned}
	f_1&=
	(4w^4+3w^3+2w^2+3w+1)x
	+3w^3+2w^2+w+2,\\[1mm]
	f_2&=
	(3w^4+w^3+4w^2+w+4)x
	+4w^4+w^3+3w^2+4w+2,\\[1mm]
	f_3&=
	(w^4+4w^3+4w)x
	+w^4+3w^3+3w^2+w,\\[1mm]
	f_4&=
	4x+2w^4+w^3+4w^2+3w,\\[1mm]
	f_5&=
	(2w^4+4w^3+3w+1)x
	+4w^4+2w^3+2w^2+2w+4.
\end{aligned}
\]
The MDS skew $(f,\theta)$-polycyclic codes of length $10$ generated by the factors of $x^{10}+x^{5}+1$ over $\mathbb{F}_{7^7}$, as listed in \Cref{tab:x10x51codes}.
\begin{table}[ht]
	\centering

	\renewcommand{\arraystretch}{1.4}
	\begin{tabular}{|c|c|c|}
		\hline
		Generator polynomial & $[n,k,d]$ & MDS \\
		\hline
		
		$x^2+f_1$
		& $[10,8,3]$
		& Yes \\
		\hline
		
		$(x^2+f_1)(x^2+f_2)$
		& $[10,6,5]$
		& Yes \\
		\hline
		
	$(x^2+f_1)(x^2+f_2)(x^2+f_3)$
		& $[10,4,7]$
		& Yes \\
		\hline
		$(x^2+f_3)(x^2+f_4)(x^2+f_5)$
		& $[10,4,7]$
		& Yes \\
		\hline
		
	\end{tabular}
		\caption{Parameters of skew polycyclic codes generated by factors of $x^{10}+x^5+1$.}
	\label{tab:x10x51codes}
\end{table}

\begin{remark}
When $\theta$ is identity and $g(x)=(x^2+f_3)(x^2+f_4)(x^2+f_5)$ where $f_3, f_4, f_5$ as defined above.	Then the  polycyclic code generated by $g(x)$ has parameters $[10,4,6]$, whereas the corresponding skew polycyclic code in \Cref*{tab:x10x51codes} has parameters $[10,4,7]$. Hence, the skew structure induced by the automorphism
	$\theta(a)=a^5$ improves the minimum distance in this case.
	
	On the other hand, consider the polynomial $	x^{10}-x^5+1$
		over \(\mathbb{F}_{5^5}\), where \(\theta(a)=a^5\) for all \(a\in \mathbb{F}_{5^5}\) and let  $\omega$ be a $(5^5-1)^{th}$ primitive root of an irreducible polynomial of degree $5$ in $\mathbb{F}_5[x]$. 
	In view of \Cref{sub 5.1}, let
	$
	f(x)=
	(x^2+a_1)(x^2+a_2)(x^2+a_3),
	$
	where
	\[	\begin{aligned}
		a_1&=
		(4w^3+4w+1)x
		+2w^4+3w^3+w^2+3w+3,\\
		a_2&=
		(2w^4+4w^3+4w^2+2w)x
		+4w^4+w^3+4w^2+2w+4,\\
		a_3&=
		(4w^4+2w^3+2w^2+3)x
		+4w^4+3w^3+2w^2+3w.
	\end{aligned}
	\]
	
	The skew polycyclic code of length \(10\) generated by \(f(x)\) over 
	\(\mathbb{F}_{5^5}[x;\theta]\) has parameters \([10,4,6]\), whereas the corresponding polycyclic code generated by the same polynomial in  \(\mathbb{F}_{5^5}[x]\) has parameters \([10,4,7]\). 
	
	Therefore, these examples show that the effect of the skew structure on the minimum distance is not uniform; depending on the defining polynomial, the skew polycyclic code may have either better or worse minimum distance than the corresponding polycyclic code.
\end{remark}

			\bibliographystyle{amsalpha}
			\bibliography{SkewConsta}

		\end{document}